\documentclass[11pt]{article}
\usepackage{mathpazo}
\usepackage[T1]{fontenc}
\usepackage[latin9]{inputenc}
\usepackage{geometry}
\usepackage{babel}
\usepackage{float}
\usepackage{url}
\usepackage{amsmath}
\usepackage{amsthm}
\usepackage{amssymb}
\usepackage{undertilde}
\usepackage{xspace}
\usepackage{microtype}
\usepackage{enumitem}
\usepackage{booktabs}

\usepackage[unicode=true,pdfusetitle, bookmarks=true,bookmarksnumbered=false,bookmarksopen=false,breaklinks=true,pdfborder={0 0 1},backref=false,colorlinks=true, citecolor=blue] {hyperref}

\makeatletter

\usepackage{bm}
\usepackage{extarrows}
\usepackage{mathtools}
\usepackage{bbm}
\usepackage{commath}
\usepackage{algorithm}
\usepackage{subcaption}
\usepackage{algorithmic}
\DeclarePairedDelimiter\ceil{\lceil}{\rceil}

\newtheorem{theorem}{Theorem}

\newtheorem{lemma}{Lemma}
\newtheorem{assumption}{Assumption}

\usepackage{appendix,breakurl}
\PassOptionsToPackage{hyphens}{url}

\makeatother

\begin{document}

\title{RL-QN: A Reinforcement Learning Framework for Optimal Control of Queueing Systems}

\author{Bai Liu $^\dagger$, Qiaomin Xie$^\ddagger$, Eytan Modiano $^\dagger$ \\
\normalsize $^\dagger$ Massachusetts Institute of Technology\\
\normalsize  $^\ddagger$ University of Wisconsin-Madison
}

\date{}                     

\maketitle

\begin{abstract}
  With the rapid advance of information technology, network systems have become increasingly complex and hence the underlying system dynamics are often unknown or difficult to characterize. Finding a good network control policy is of significant importance to achieve desirable network performance (e.g., high throughput or low delay). In this work, we consider using model-based reinforcement learning (RL) to learn the optimal control policy for queueing networks so that the average job delay (or equivalently the average queue backlog) is minimized. Traditional approaches in RL, however, cannot handle the unbounded state spaces  of the network control problem. To overcome this difficulty, we propose a new algorithm, called Reinforcement Learning for Queueing Networks (RL-QN), which applies model-based RL methods over a finite subset of the state space, while applying a known stabilizing policy for the rest of the states. We establish that the average queue backlog under RL-QN with an appropriately constructed subset can be arbitrarily close to the optimal result. We evaluate RL-QN in dynamic server allocation, routing and switching problems. Simulation results show that RL-QN minimizes the average queue backlog effectively.
\end{abstract}


\section{Introduction} \label{Sec:Intro}

The rapid growth of information technology has resulted in increasingly complex network systems and poses challenges in obtaining explicit knowledge of system dynamics. For instance, due to security or economic concerns, a number of network systems are built as overlay networks, e.g. caching overlays, routing overlays and security overlays \cite{sitaraman2014overlay}. In these cases, only the overlay part is fully controllable by the network administrator, while the underlay part remains uncontrollable and/or unobservable. The ``black box" components make network control policy design challenging.

In addition to the challenges brought by unknown system dynamics, many of the current network control algorithms (e.g., MaxWeight~\cite{tassiulas1990stability} and Drift-plus-Penalty~\cite{neely2008fairness}) aim at stabilizing the system. However, existing optimization methods for long-term performance metrics (e.g. queue backlog and delay) only work for specific structures \cite{buyukkoc1985cmu, maguluri2016heavy}, while the development for general setting remains insufficient.

To overcome the above challenges, it is desirable to apply inference and learning schemes. A natural approach is reinforcement learning, which learns and reinforces a good decision policy by interacting with the environment and receiving feedbacks.
Reinforcement learning methods provide a framework for the design of learning policies for general networks. There have been two main lines of work on reinforcement learning methods: model-free reinforcement learning (e.g., Q-learning \cite{watkins1992q}, policy gradient \cite{sutton2000policy}) and model-based reinforcement learning (e.g., UCRL \cite{jaksch2010near}, PSRL \cite{osband2013more}). In this work, we focus on the model-based approach. 



\subsection{Related Work} \label{Sec:RelatedWork}

Existing methods for reducing the average job delay (or equivalently reducing the average queue backlog) of general networks can be classified into three types: equivalent constraint, Lyapunov drift and Markov decision process (MDP) \cite{cui2012survey}. However, most equivalent constraint approaches can only be applied to single-hop networks \cite{cui2012survey}. Therefore, we focus on discussing the latter two classes of algorithms, as they are more universal and can be applied to multi-hop network systems.

We first discuss the Lyapunov drift approach, which generally does not require learning the dynamics for network problems with hidden dynamics. The Lyapunov drift approach has been widely applied in numerous network control problems. For instance, in dynamic server allocation problems, a Lyapunov-drift-based algorithm named the longest connected queue (LCQ) can stabilize the queue backlog without knowing the underlying dynamics (e.g., arrival rates, channel statistics) \cite{tassiulas1993dynamic,ganti2007optimal}. For multiclass routing networks, the Backpressure algorithm was designed under the Lyapunov drift framework and only requires the observation of queue backlogs \cite{tassiulas1990stability,awerbuch1993simple,neely2003dynamic}. Similarly, Lyapunov drift methods that do not require learning the network dynamics have been proposed for switch scheduling \cite{mekkittikul1998practical,keslassy2001analysis,shah2012switched} and inventory control problems \cite{neely2010dynamic, dai2005maximum} etc. Most Lyapunov drift algorithms are shown to be throughput optimal \cite{tassiulas1990stability}, i.e., they can stabilize the system whenever the system is stabilizable. Also, the Lyapunov drift approach usually only requires solving a linear programming problem and does not suffer from the ``curse of dimensionality'', which makes it applicable to large-scale queueing systems. However, without learning the system dynamics, the Lyapunov drift approach generally cannot guarantee minimum queueing delay. For instance, for dynamic server allocation problems, the optimal policy has only been developed for highly symmetric systems (i.e., uniform external arrival rates, uniform connectivity probabilities and uniform success rates) \cite{ganti2007optimal}. Another well-known example is switching system, for which the Maximum Matching policy has been shown to be close to the optimal policy when the system is in the heavy traffic regime~\cite{maguluri2016heavy}. In Section 5.1, 5.2 and 5.4, we conduct numerical experiments and show that our approach significantly outperforms the Lyapunov drift methods.

The MDP approach models the queueing system control problem as an MDP that aims to minimize the long-term average queue backlog. Classical algorithms to solve MDP problems include value iteration and policy iteration \cite{bertsekas2017dynamic}. Note that although the MDP approach can minimize the average job delay, it can only be applied to networks with finite buffer size, which is unrealistic for many practical network models. Therefore, the MDP approach fails to minimize the average job delay of general stochastic networks and is typically applied to MDPs of small scale \cite{yeh2001multiaccess,yeh2003throughput} or special structures \cite{guestrin2003efficient, qu2020scalable}. However, traditional MDP approach requires explicit knowledge of the network dynamics. Since the network problems studied in this work have hidden dynamics, it is necessary to learn the dynamics when applying the MDP approach. A related field is reinforcement learning. The majority of reinforcement learning algorithms apply approximations to reduce the computation complexity, including state representation \cite{de2018integrating,maillard2013selecting}, value function approximation \cite{baird1995residual,jin2020provably} and pruning \cite{livne2020pops}. Reinforcement learning methods have been applied to a wide range of network control problems, like routing \cite{boyan1994packet,peshkin2002reinforcement,chen2018deep,yu2018drom}, spectrum access \cite{naparstek2017deep,wang2018deep,guan2020robust}, switch scheduling \cite{brown2001switch}, state probing \cite{brun2016big} etc. Although these reinforcement learning methods achieve satisfactory performance in simulation, they lack performance guarantees.

We consider model-based reinforcement learning methods, which tend to be more analytically tractable. Conventional model-based reinforcement learning methods like UCRL \cite{jaksch2010near} and PSRL \cite{osband2013more} only work for finite-state-space systems, yet queueing systems are usually modeled to have unbounded buffer sizes. The work in \cite{liang2019optimal} assumes that the system has an admission control scheme to keep queue backlogs finite. The resulting system can be model as a MDP with finite states, where UCRL can be applied. In \cite{theocharous2017posterior}, the authors modify the PSRL algorithm to deal with MDPs with large state space, yet the algorithm requires the MDP to have a finite bias span, which is unrealistic for problems that aim to minimize the average cost with a countably infinite state space.

In summary, existing methods either cannot guarantee optimal delay (the Lyapunov drift approach), or lack performance guarantees (heuristic reinforcement learning) or can only deal with finite state spaces (the MDP approach and theoretical reinforcement learning). In this paper, we aim to develop an algorithm that achieves provable optimality for networks with countably infinite state spaces. To the best of our knowledge, our algorithm is the first to achieve minimum average queue backlog in general networks with unbounded buffers. Our analysis also offers a possible roadmap to solving general MDPs with countably infinite state spaces.


\subsection{Our Contributions} \label{Sec:Contribution}

We apply model-based reinforcement learning to queueing networks with unbounded state spaces and unknown dynamics. Our approach leverages the fact that for a vast class of stable queueing systems, the probability of the queue backlog being large is relatively small. This observation motivates us to focus on learning control policies over a finite subset of states that the system visits with high probability. Our main contributions are summarized as follows. 

\begin{itemize}
    \item We propose a model-based RL algorithm that can deal with unbounded state spaces. In particular, we introduce an auxiliary system with the state space bounded by a threshold $U$. Our approach employs a piecewise policy: for states below the threshold, a model-based RL algorithm is used; for all other states, a simple baseline algorithm is applied.
    
    \item We establish that the episodic average queue backlog under the proposed algorithm can be arbitrarily close to the optimum with a large threshold $U$. In particular, by applying Lyapunov analysis, we characterize the gap to the optimal performance as a function of the threshold $U$. In addition, our proof technique may be of independent interest for analyzing the convergence of other reinforcement learning algorithms for queueing networks.
    
    \item Simulation results on dynamic server allocation, routing, and network switching problems corroborate the validity of our theoretical guarantees. In particular, the proposed algorithm effectively achieves a small average queue backlog, with the gap to optimum diminishing as the threshold $U$ increases.
\end{itemize}

The paper is organized as follows. In Section \ref{Sec:Prob} we formulate the problem and introduce notations. Section \ref{Sec:Approach} gives an outline of our approach, presents required assumptions and the proposed algorithm. We conduct theoretical performance analysis regarding convergence and optimality of the proposed algorithm in Section \ref{Sec:Theory}. We evaluate the algorithm in various settings and the numerical results are given in Section \ref{Sec:Sim}.


\section{System Model} \label{Sec:Prob}

We consider a discrete time network with the general topology of a directed graph $\mathcal{G} = (\mathcal{N}, \mathcal{L})$, where $\mathcal{N}$ is the set of nodes and $\mathcal{L}$ is the set of links. Each node maintains one or more queues for undelivered packets, and each queue has an unbounded buffer. We denote by $D$ the number of queues. Nodes represent the locations where data packets are generated, relayed and/or processed in the network. Links represent the communication links.

During each time slot, external data packets arrive at nodes where they are either processed and then depart the system or are relayed. For relayed packets, the communication links can be stochastic, i.e. data transmissions between nodes can fail. We assume the underlying dynamics are partially or fully unknown. This model captures a large class of queueing networks that involve routing, scheduling and switching.

We suppose that the system has the Markovian property: the probability distributions of the queue backlogs in the next time slot only depend on the current queue backlogs and the control decision. In other words, the system can be modeled as a countable-state-space MDP $\mathcal{M}$ with average cost as follows. 

\begin{itemize}
\item State space $\mathcal{S}$: the set of queue backlog vectors, i.e., $\mathcal{S} \triangleq \mathbb{N}^D$.
\item Action space $\mathcal{A}$: the set of feasible control decisions, which are specified by the problem setting. In this paper we only consider finite action space.
\item State-transition probability $p \left( \bm{Q}' \mid \bm{Q}, a \right)$: the probability of transitioning into state $\bm{Q}'$ from state $\bm{Q}$ with action $a$. For simplicity, we assume that the magnitude that a queue backlog can change during one time slot is upper bounded by a constant $W$. Let $\mathcal{R} \left( \bm{Q} \right)$ denote the set of one-step reachable queue backlog vectors from state $\bm{Q}$.
\item Cost function $c \left( \bm{Q} \right)$: the total backlog of $D$ queues, i.e., $c \left( \bm{Q} \right) \triangleq \sum_{i=1}^{D} Q_i$.
\end{itemize}

As discussed in Section \ref{Sec:RelatedWork}, a large number of queueing networks have been studied extensively and various stabilizing policies that aim to keep the average total queue backlog finite have been developed. Here we take a step further to go beyond stability and our goal is to minimize the average queue backlog.

For readers' convenience, we summarize the notations used in this paper in Table \ref{Tab:Notations}.

\begin{table}[htbp]\caption{Notations}
\begin{center}
\begin{tabular}{r p{12cm} }
\toprule
$D$ & The number of queues in the queueing network \\
$\bm{Q}$ & The $D$-dimensional queue backlog vector of the queueing network \\
$Q_i$ & The queue backlog of node $i$ \\
$Q_{\max}$ & The maximum entry of $\bm{Q}$ \\
$\mathcal{R} \left( \bm{Q} \right)$ & The set of one-step reachable queue backlog vectors from $\bm{Q}$ \\
$W$ & The magnitude that a queue backlog can change during one time slot \\
$U$ & The buffer size of the auxiliary system \\
$\mathcal{M}$ & The countable-state-space MDP of the original queueing network \\
$\tilde{\mathcal{M}}$ & The MDP of the auxiliary system \\
$\mathcal{S}$ & The state space of $\mathcal{M}$ \\
$\tilde{\mathcal{S}}$ & The state space of $\tilde{\mathcal{M}}$ \\
$\mathcal{A}$ & The action space of $\mathcal{M}$ \\
$p \left( \cdot \mid \cdot, \cdot \right)$ & The state-transition kernel of $\mathcal{M}$  \\
$\tilde{p} \left( \cdot \mid \cdot, \cdot \right)$ & The state-transition kernel of $\tilde{\mathcal{M}}$  \\
$\pi_0$ & The known stabilizing policy of the queueing network \\
$\pi_{\mathrm{rand}}$ & The random policy that selects an action in $\mathcal{A}$ uniformly \\
$\pi^*$ & The stationary policy that minimizes the expected average total queue backlog of $\mathcal{M}$ \\
$\tilde{\pi}^*$ & The stationary policy that minimizes the expected average total queue backlog of $\tilde{\mathcal{M}}$ \\
$\rho^*$ & The expected average queue backlog of $\mathcal{M}$ when applying $\pi^*$ \\
$\tilde{\rho}^*$ & The expected average queue backlog of $\tilde{\mathcal{M}}$ when applying $\tilde{\pi}^*$ \\
$\tilde{\Phi}^* (\cdot)$ & The Lyapunov function with negative drift regarding $\tilde{\pi}^*$ \\
$\beta$ & The order of $\tilde{\Phi}^* (\cdot)$ \\
$\mathcal{S}^{\mathrm{in}}$ & The set of $\bm{Q}$'s with $\tilde{\Phi}^* \left( \bm{Q} \right) \leqslant (U-W)^\beta$ \\
$\mathcal{S}^{\mathrm{out}}$ & $\mathcal{S} \setminus \mathcal{S}^{\mathrm{in}}$ \\
$p^{\pi + \pi'} (\cdot)$ & The stationary probability distribution in $\mathcal{M}$ with $\pi$ applied to $\mathcal{S}^{\mathrm{in}}$ and $\pi'$ applied to $\mathcal{S}^{\mathrm{out}}$ \\
$\mathbb{E}_{\pi} [\cdot]$ & The expection of a random variable in $\mathcal{M}$ with $\pi$ applied to $\mathcal{S}$ \\
$\mathbb{E}_{\pi + \pi'} [\cdot]$ & The expection of a random variable in $\mathcal{M}$ with $\pi$ applied to $\mathcal{S}^{\mathrm{in}}$ and $\pi'$ applied to $\mathcal{S}^{\mathrm{out}}$ \\
$\tilde{\mathbb{E}}_{\tilde{\pi}} [\cdot]$ & The expection of a random variable in $\tilde{\mathcal{M}}$ with $\tilde{\pi}$ applied to $\tilde{\mathcal{S}}$ \\

\bottomrule
\end{tabular}
\end{center}
\label{Tab:Notations}
\end{table}


\section{Our Approach} \label{Sec:Approach}

Classical model-based reinforcement learning fits a model of state transition kernel to observed data and then solves the dynamic programming problem on the estimated system. The challenge of applying such an approach to countable-state MDPs arises from both of estimating model parameters and solving the estimated MDP, due to the fact that the state space is unbounded. 


Here we introduce an auxiliary system $\tilde{\mathcal{M}}$ with a bounded state space. We only apply reinforcement learning techniques on the constructed bounded state space, while simply apply a \emph{known} stabilizing policy $\pi_0$ (cf. Assumption~\ref{Asp:KnownPolicy}) to the rest of the states. We show that the performance gap between the proposed algorithm and the optimal policy can be made arbitrarily small by designing appropriate $\tilde{\mathcal{M}}$.



\subsection{Overview} \label{Sec:Outline}

We first provide an overview of our approach and defer a detailed description of the algorithm to Section~\ref{Sec:Algorithm}. Our reinforcement learning method operates in an episodic manner. 

We apply a decaying $\epsilon$-greedy method to decide whether an episode should conduct exploration or exploitation. For each episode, we perform exploration (i.e., apply some random policy to collect diverse samples) with probability $\epsilon_i$, and we conduct exploitation (i.e., using current estimates of dynamics to compute an estimated optimal policy) with probability $1 - \epsilon_i$. At the beginning of the learning process, we tend to explore the system to collect samples, and thus $\epsilon_i$ is relatively large when $i$ is small. As the learning process goes on, we gradually obtain enough samples and are close to the optimal policy, and thus $\epsilon_i$ is decreased to exploit the learned policy. The scheme has been applied in network control \cite{alaya2008dynamic,sodagari2011anti,guan2020robust} to achieve a trade-off between exploration and exploitation.

For an exploration episode, we apply a randomized policy $\pi_{\mathrm{rand}}$ that takes action uniformly to obtain samples for the estimation of state-transition probabilities in $\mathcal{M}$. However, since $\mathcal{M}$ has a countably infinite state space, we instead estimate the model for an auxiliary system $\tilde{\mathcal{M}}$ with a bounded state space. The auxiliary system $\tilde{\mathcal{M}}$ has a threshold $U$: the system has the same dynamics as the real one, with the only difference that each queue has a bounded buffer size $U$. For each queue in $\tilde{\mathcal{M}}$, when its queue backlog reaches $U$, new packets to it will be dropped. Mathematically, the state space of $\tilde{\mathcal{M}}$ is given by $\tilde{\mathcal{S}} \triangleq \{ \bm{Q} \in \mathcal{S} : Q_{\max} \leqslant U \}$ where $Q_{\max} \triangleq \max Q_i$. The auxiliary system $\tilde{\mathcal{M}}$ shares the same action space $\mathcal{A}$ and cost function $c \left( \bm{Q} \right)$ as $\mathcal{M}$. With the introduction of $\tilde{\mathcal{M}}$, an exploration episode operates as in Figure \ref{Fig:Approach_1}: we only apply $\pi_{\mathrm{rand}}$ to states in $\tilde{\mathcal{S}}$, while simply applying the \emph{known} stabilizing policy $\pi_0$ to other states.

For an exploitation episode, we first repartition the state space $\mathcal{S}$ due to technical reasons. We denote by $\tilde{\pi}^*$ the optimal policy for $\tilde{\mathcal{M}}$, and define a Lyanopuv function $\tilde{\Phi}^* (\cdot)$ (cf. Assumption~\ref{Asp:DriftOptimalPolicy}). Denoting by $\beta$ the order of $\tilde{\Phi}^* (\cdot)$, we partition $\mathcal{S}$ in the following manner:
$$
\begin{cases}
\mathcal{S}^{\mathrm{in}} \triangleq \left\{ \bm{Q} : \tilde{\Phi}^* \left( \bm{Q} \right) \leqslant (U-W)^\beta \right\} \\
\mathcal{S}^{\mathrm{out}} \triangleq \mathcal{S} \setminus \mathcal{S}^{\mathrm{in}}
\end{cases}
.
$$
As illustrated in Figure \ref{Fig:Approach_2}, during an exploitation episode, we first compute an estimated $\tilde{\pi}^*$ using the estimated dynamics obtained from exploration episodes. We then apply the estimated $\tilde{\pi}^*$ to states in $\mathcal{S}^{\mathrm{in}}$ and $\pi_0$ to states in $\mathcal{S}^{\mathrm{out}}$ throughout the episode.

\begin{figure}
  \centering
  \begin{subfigure}[b]{0.47\textwidth}
      \centering
      \includegraphics[width=\linewidth]{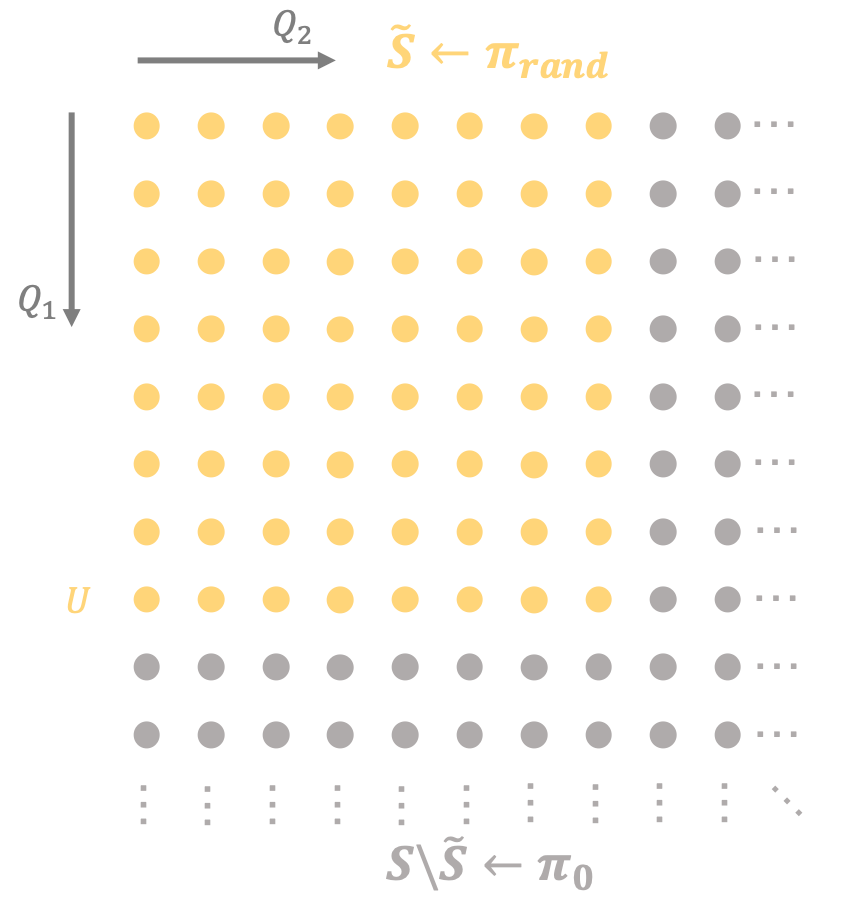}
      \caption{Exploration}
      \label{Fig:Approach_1}
  \end{subfigure}%
  \hfill
  \begin{subfigure}[b]{0.47\textwidth}
      \centering
      \includegraphics[width=\linewidth]{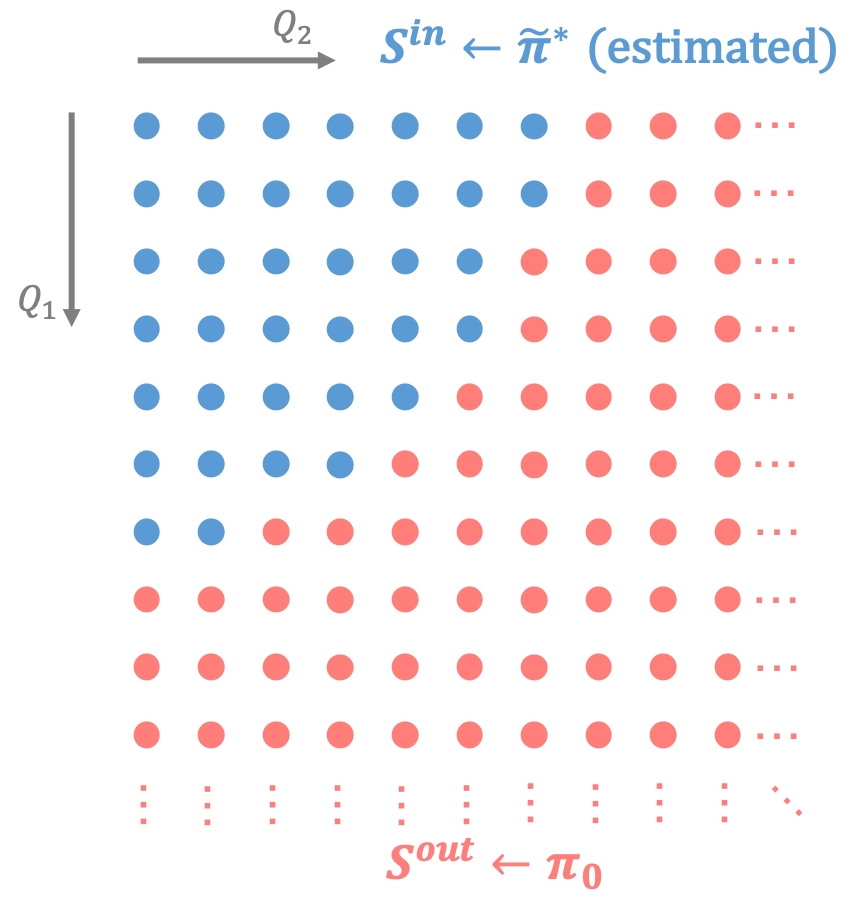}
      \caption{Exploitation}
      \label{Fig:Approach_2}
  \end{subfigure}
  \caption{Schemetic illustration of our approach (when $D=2$).}
  \label{Fig:Approach}
\end{figure}

We will show that the average queue backlog under our algorithm  converges to the optimal average queue backlog $\rho^*$ as we increase $U$. We divide the analysis into two stages: before and after $\tilde{\pi}^*$ is learned. For the first stage, our model-based reinforcement learning approach applies $\epsilon$-greedy exploration. We show that the proposed algorithm gradually obtains $\tilde{\pi}^*$ with high probability. For the second stage, by applying drift analysis on the Markov chain, we show that when $\tilde{\pi}^*$ is applied for states in $\mathcal{S}^{\mathrm{in}}$ and $\pi_0$ for states in $\mathcal{S}^{\mathrm{out}}$, the probability of queue backlog exceeding into $\mathcal{S}^{\mathrm{out}}$ decays exponentially with $U$. In addition, whenever $\bm{Q}$ leaves $\mathcal{S}^{\mathrm{in}}$, the expected accumulated queue backlog before $\bm{Q}$ returns back to $\mathcal{S}^{\mathrm{in}}$ can be upper bounded as a polynomial term in $U$. Together, we show that the gap between our result and the optimal average queue backlog $\rho^*$ is upper bounded by $\mathcal{O} \left( \mathrm{poly}(U) / \exp \left( U \right) \right)$, which diminishes as $U$ increases.


\subsection{Assumptions}


Our algorithm can be applied to a broad class of network problems. To establish rigorous performance guarantee, we need to impose some assumptions on $\mathcal{M}$ and $\tilde{\mathcal{M}}$. We will discuss these assumptions and argue that they are reasonable under many queueing networks.

\medskip

\noindent \textbf{Assumption on $\pi_0$:} As introduced in Section~\ref{Sec:Outline}, we only learn the optimal policy for an auxiliary system with a bounded state space, while for the rest of the states, we apply a \emph{known} stabilizing policy $\pi_0$ to control the queue backlog. We  quantify the stability of $\pi_0$ as follows.
\begin{assumption} \label{Asp:KnownPolicy}
There exists a \emph{known} policy $\pi_0$, a Lyapunov function $\Phi_0: \mathcal{S} \rightarrow \mathbb{R}_{+}$ and constants $a, \alpha, \epsilon_0, B_0 > 0$ such that the following properties hold:
\begin{enumerate}
    \item $\Phi_0 \left( \bm{Q} \right) \leqslant a Q_{\max}^{\alpha}$ for each $\bm{Q} \in \mathcal{S}$ ;
\item $\mathbb{E}_{\pi_0} \left[\Phi_0 \left( \bm{Q}(t+1) \right) - \Phi_0 \left( \bm{Q}(t) \right) \mid \bm{Q}(t) = \bm{Q} \right] \leqslant -\epsilon_0$ for each  $\bm{Q}$ such that $Q_{\max} \geqslant B_0$.
\end{enumerate}

\end{assumption}

The requirement of MDP respecting Lyapunov function under a stable policy is not restrictive. For a stabilizable stochastic network, there exists a policy $\pi_0$ under which the corresponding Markov chain is positive recurrent. Moreover, the property of positive recurrence is known to be equivalent to the existence of the so-called Lyapunov function~\cite{mertens1978necessary}. These observations motivate us to focus on MDPs that satisfy Assumption \ref{Asp:KnownPolicy}. In practice, for stabilizable networks, numerous stable policies have been proposed, including dynamic server allocation problems \cite{baras1985k, buyukkoc1985cmu, tassiulas1993dynamic}, multiclass routing networks \cite{bertsimas1994optimization, humes1997delay, kumar1994performance, kumar1995stability, bertsimas2001performance}, switch scheduling \cite{shah2012switched,keslassy2001analysis,tassiulas1993dynamic} and inventory control \cite{neely2010dynamic, dai2005maximum}, which all satisfy Assumption \ref{Asp:KnownPolicy}.

\medskip

\noindent\textbf{Assumption on $\tilde{\pi}^*$: } From the outline of performance analysis in Section~\ref{Sec:Outline}, we need to show that the probability of queue backlog entering $\mathcal{S}^{\mathrm{out}}$ decays exponentially with respect to $U$. We therefore make a natural assumption on the stability property of $\tilde{\pi}^*$. This assumption is necessary for the proof of Lemma \ref{Lem:pBd}. For clarity, we use $\tilde{\mathbb{E}}_{\tilde{\pi}}[\cdot]$ to denote the expectation with respect to the randomness of the auxilary system $\tilde{\mathcal{M}}$ under the policy $\tilde{\pi}$ (to distinguish from $\mathbb{E}[\cdot]$ for $\mathcal{M}$).
\begin{assumption} \label{Asp:DriftOptimalPolicy}
There exist a Lyapunov function $\tilde{\Phi}^*: \tilde{\mathcal{S}}\rightarrow \mathbb{R}_+$ and constants $\beta \geqslant 1$, $b_1, b_2, \tilde{B}^*, \tilde{\epsilon}^* > 0$, such that for any $U > 0$, the following properties hold:
\begin{enumerate}
\item $Q_{\max}^{\beta} \leqslant \tilde{\Phi}^* \left( \bm{Q} \right) \leqslant b_1 Q_{\max}^{\beta}$ for each $\bm{Q} \in \tilde{\mathcal{S}}$ ;
\item $\max_{\bm{Q}' \in \mathcal{R} \left( \bm{Q} \right)} \abs{\tilde{\Phi}^* \left( \bm{Q}' \right) - \tilde{\Phi}^* \left( \bm{Q} \right)} \leqslant b_2 \left( Q_{\max}^{\beta-1} + Q_{\max}^{'\beta-1} \right)$ for each $\bm{Q} \in \tilde{\mathcal{S}}$ ;
\item $\tilde{\mathbb{E}}_{\tilde{\pi}^*} \left[\tilde{\Phi}^* \left( \bm{Q}(t+1) \right) - \tilde{\Phi}^* \left( \bm{Q}(t) \right) \mid \bm{Q}(t) = \bm{Q} \right] \leqslant -\tilde{\epsilon}^* \cdot Q_{\max}^{\beta-1}$ for each $\bm{Q}\in \tilde{\mathcal{S}}$ such that $Q_{\max} \geqslant \tilde{B}^*$.
\end{enumerate}
\end{assumption}

The assumption is natural in the sense that optimal policies are stable and thus are likely to have good Lyapunov drift properties. For instance, if $\tilde{\Phi}^* \left( \bm{Q} \right) = \sum_{i=1}^D \omega_i Q_i^2$ has a negative drift for each $\bm{Q}$ satisfying $\sum_{i=1}^D Q_i \geqslant B$, one can easily show that $\omega_{min} \cdot Q_{\max}^2 \leqslant \tilde{\Phi}^* \left( \bm{Q} \right) \leqslant(\sum_{i=1}^D \omega_i) \cdot Q_{\max}^2$, $\max_{\bm{Q}' \in \mathcal{R} \left( \bm{Q} \right)} \abs{\tilde{\Phi}^* \left( \bm{Q}' \right) - \tilde{\Phi}^* \left( \bm{Q} \right)} \leqslant W \cdot(\sum_{i=1}^D \omega_i) \cdot (Q_{\max} + Q'_{\max})$, and has a negative drift for each $\bm{Q}$ with $Q_{\max} \geqslant B$. Possible methods to analyze Lyapunov drift properties from queueing stability can be found in \cite{meyn2012markov}. 

\medskip

\noindent\textbf{Assumption on communication properties: } Existing model-based reinforcement learning algorithms with performance guarantees usually require the ``diameter'' (i.e. the upper bound for the shortest first hitting time between any two states) of the MDP to be finite \cite{jaksch2010near}. However, it is unrealistic to assume the MDP diameter to be finite when the state space is unbounded. Instead, we make the following assumption (this assumption is used in the proof of Lemma \ref{Lem:EpiExBacklog}) where $T_{\bm{Q} \to \bm{Q}'}$ denotes the first hitting time from $\bm{Q}$ to $\bm{Q}'$.

\begin{assumption} \label{Asp:PolyTrans}
There exist constants $c, \gamma > 0$, such that for any $U > 0$ and every $\bm{Q}, \bm{Q}' \in \tilde{\mathcal{S}}$,
$$
\min_{\tilde{\pi}} \tilde{\mathbb{E}}_{\tilde{\pi}} \left[ T_{\bm{Q} \to \bm{Q}'} \right] \leqslant c \lVert \bm{Q}' - \bm{Q} \rVert_1^{\gamma} ,
$$
where $\tilde{\pi}$ is a policy that can be applied to $\tilde{S}$.
\end{assumption}

In other words, Assumption~\ref{Asp:PolyTrans} states that there exists a policy $\tilde{\pi}$ such that the first hitting time between two states $\bm{Q}$ and $\bm{Q}'$ in $\tilde{\mathcal{S}}$ is a polynomial function of $\lVert \bm{Q}' - \bm{Q} \rVert_1$. We emphasize that the constant $\gamma$ is independent of the value of $U$. Indeed, with larger $U$'s, $\tilde{\mathbb{E}}_{\tilde{\pi}} \left[ T_{\bm{Q} \to \bm{Q}'} \right]$ might grow larger, since the trajectory from $\bm{Q}$ to $\bm{Q}'$ are more likely to be longer. However, it is possible that the increment is bounded. For instance, certain $M/G/1$ queueing systems \cite{ross1999hitting} and random walk models \cite{khantha1983first} have constant $\gamma$'s even if the state space is unbounded. Directly verifying Assumption~\ref{Asp:PolyTrans} might be computationally difficult. Theorem 11.3.11 from \cite{meyn2012markov} offers a drift analysis approach for justifying Assumption~\ref{Asp:PolyTrans}. 

We remark that the maximum first passage time of the MDP plays an essential role in analyzing the reinforcement learning algorithms in infinite-horizon average-reward MDPs. Existing related works either use MDP diameter \cite{jaksch2010near,zhang2019regret}, or span \cite{bartlett2012regal,ouyang2017learning,fruit2018efficient}, or mixing time \cite{ortner2020regret,abbasi2019politex,wei2020model} to characterize the maximum first passage time of the MDP, where these metrics serve as coefficients in the gap to the optimum. The analysis in our paper also relies on
the assumption of reasonable bounds for the maximum first passage time. 
It would be of great interest to establish performance guarantees under relaxed assumptions.

\medskip

\noindent \textbf{Assumption on error tolerance: } As the learning process proceeds, ideally the estimation for $\tilde{\mathcal{M}}$ becomes increasingly accurate. However, it is unrealistic for us to obtain the exact $\tilde{\mathcal{M}}$. Therefore, we assume that if the estimates for the state-transition kernels are accurately enough (i.e., within a certain error bound), the optimal policy to the estimated $\tilde{\mathcal{M}}$ is the same as $\tilde{\pi}^*$. The assumption is stated as follows.
\begin{assumption} \label{Asp:ErrDist}
There exists a $\Delta p > 0$, such that for any MDP $\tilde{\mathcal{M}}'$ with the same state space, action space and cost function as $\tilde{\mathcal{M}}$, if for each $\bm{Q} \in \tilde{\mathcal{S}}$ and each $a \in \mathcal{A}$, we have
$$
\norm{ \tilde{p} \left( \cdot \mid \bm{Q}, a \right) - \tilde{p}' \left( \cdot \mid \bm{Q}, a \right) }_1 \leqslant \Delta p ,
$$
then the optimal policy for $\tilde{\mathcal{M}}'$ is the same as the optimal policy $\tilde{\pi}^*$ for $\tilde{\mathcal{M}}$.
\end{assumption}

In this paper we focus on network control problems with finite action space. For many networks, when the dynamics (e.g. exogenous arrival rates, service rates, channel capacities) vary slightly, the optimal policies remain the same. For instance, if the arrival rates of the switching networks vary yet remain heavy loads, Maximum Matching still remains a close-to-optimal policy \cite{maguluri2016heavy}. 


\subsection{Algorithm} \label{Sec:Algorithm}

We propose an algorithm called Reinforcement Learning for Queueing Networks (RL-QN). 
RL-QN operates in an episodic manner: at the beginning of episode $k$, we uniformly draw a real number $\xi \in [0, 1]$ to decide whether to explore or exploit during this episode. The length of each episode depends on the observations.
\begin{itemize}
    \item If $\xi \leqslant \epsilon_k \triangleq \ell / \sqrt{k} $ (where $\ell \in (0,1]$ can be tuned to control the exploration frequency), we perform exploration during this episode. For states in $\tilde{\mathcal{S}}$, we apply the randomized policy $\pi_{\mathrm{rand}}$. For states in $\mathcal{S} \setminus \tilde{\mathcal{S}}$, we apply $\pi_0$. 
    
    \item If $\xi > \epsilon_k$, we enter the exploitation stage. We first calculate sample means to estimate the state-transition function $\tilde{p}$ of $\tilde{\mathcal{M}}$. We then apply value iteration on the estimated system $\tilde{M}_k$ to obtain an estimated optimal policy $\tilde{\pi}^*_k$.  During this episode, we apply $\tilde{\pi}^*_k$ for states in $\mathcal{S}^{\mathrm{in}}$ and $\pi_0$ otherwise. 
    
    \item When the number of visits to states in $\mathcal{S}^{\mathrm{in}}$ exceeds $L_k = L \cdot \sqrt{k}$ (where $L$ is a positive constant and can be tuned to adjust the sampling rate), RL-QN enters episode $k+1$ and repeat the process above.
\end{itemize}

The details are presented in Algorithm \ref{Alg:RL-QN}. We use $\min \left\{ U, \bm{Q} \right\}$ to denote a vector with the $i$-th coordinate being $\min\{U,Q_i\}.$

\begin{algorithm}
\caption{The RL-QN algorithm} \label{Alg:RL-QN}
\begin{algorithmic}[1]
  \STATE \textbf{Input: } $\mathcal{A}$, $U > W$, $0 < \ell \leqslant 1$, $L > 0$, $K > 0$

  \FOR{episodes $k \leftarrow 1, 2, \cdots,K$}

    \STATE Set $L_k \leftarrow  L \cdot \sqrt{k}$, $\epsilon_k \leftarrow \ell / \sqrt{k}$ and uniformly draw $\xi \in [0, 1]$.

    \IF{$\xi \leqslant \epsilon_k$} 
      \STATE Set 
              $$
              \pi_k \left( \bm{Q} \right) = 
              \begin{cases}
              \pi_{\mathrm{rand}} \left( \bm{Q} \right) , \text{ for } \bm{Q} \in \tilde{\mathcal{S}} \\
              \pi_0 \left( \bm{Q} \right) , \text{ for } \bm{Q} \in \mathcal{S} \setminus \tilde{\mathcal{S}}
              \end{cases} .
              $$  
    \ELSE
      \STATE For each $\bm{Q}, \bm{Q}' \in \tilde{\mathcal{S}}$ and $a \in \mathcal{A}$, let $\tilde{p}_k \left( \bm{Q}' \mid \bm{Q}, a \right) = \tilde{P} \left( \bm{Q}, a, \bm{Q}' \right) / N \left( \bm{Q}, a \right)$ for $N \left( \bm{Q}, a \right) > 0$ and $\tilde{p}_k \left( \bm{Q}' \mid \bm{Q}, a \right) = 1 / \ \abs{\mathcal{R} \left( \bm{Q} \right)}$ otherwise.
      \STATE Solve the estimated MDP $\tilde{M}_k$ and obtain the estimated optimal policy $\tilde{\pi}^*_k$.
      \STATE Set 
              $$
              \pi_k \left( \bm{Q} \right) = 
              \begin{cases}
              \tilde{\pi}^*_k \left( \bm{Q} \right) , \text{ for } \bm{Q} \in \mathcal{S}^{\mathrm{in}} \\
              \pi_0 \left( \bm{Q} \right) , \text{ for } \bm{Q} \in \mathcal{S}^{\mathrm{out}}
              \end{cases} .
              $$  
    \ENDIF

    \WHILE{visits to states in $\mathcal{S}^{\mathrm{in}}$ is smaller than $L_k$}
      \STATE Take the action $a_t = \pi_k \left( \bm{Q}(t) \right)$ for the real system.
      \STATE Observe the next state $\bm{Q}(t+1)$.
      \IF{$\bm{Q}(t) \in \tilde{\mathcal{S}}$} 
        \STATE Increase $N \left( \bm{Q}(t), a_t \right)$ by $1$.
        \STATE Increase $\tilde{P} \left( \bm{Q}(t), a_t, \min \left\{ U, \bm{Q} \right\} \right)$ by $1$. 
      \ENDIF
      \STATE $t \leftarrow t+1$.
    \ENDWHILE

  \ENDFOR

  \STATE \textbf{Output: } estimated optimal policy $\tilde{\pi}^*_K$

\end{algorithmic}
\end{algorithm}


\section{Performance Analysis} \label{Sec:Theory}

We analyze the performance of our algorithm from both exploration and exploitation perspectives, under Assumptions \ref{Asp:KnownPolicy}-\ref{Asp:ErrDist}. We first prove that RL-QN can learn $\tilde{\pi}^*$ within finite episodes with high probability, which implies that RL-QN explores different states sufficiently to obtain an accurate estimation of $\tilde{\mathcal{M}}$ (cf. Theorem~\ref{Thm:LearningProcess}). We then show that RL-QN exploits the estimated optimal policy and achieves a performance close to the optimal result of $\rho^*$ (cf. Theorem~\ref{Thm:Backlog}).

In this paper, we focus on MDPs such that all states are accessible from each other under the following policies: (a) $\pi_{\mathrm{rand}}+\pi_0:$ applying $\pi_{\mathrm{rand}}$ to $\tilde{\mathcal{S}}$ and $\pi_0$ to $\mathcal{S} \setminus \tilde{\mathcal{S}}$; (b) $\tilde{\pi}^*+\pi_0:$ applying $\tilde{\pi}^*$ to $\mathcal{S}^{\mathrm{in}}$ and $\pi_0$ to $\mathcal{S}^{\mathrm{out}}$; and (c) $\pi^*:$ applying (a truncated version of) $\pi^*$ to $\tilde{\mathcal{S}}$ in $\tilde{\mathcal{M}}$. That is, the corresponding Markov chains under the above policies are irreducible.


\subsection{Convergence to the Optimal Policy} \label{Sec:Conv2Policy}

The following theorem states that, with arbitrarily high probability, RL-QN learns $\tilde{\pi}^*$ within a finite number of episodes.

\begin{theorem} \label{Thm:LearningProcess} Suppose Assumption~\ref{Asp:ErrDist} holds. For each $\delta\in(0,1)$, there exists $k^* < \infty$ such that \emph{RL-QN} learns $\tilde{\pi}^*$ (i.e. $\tilde{\pi}^*_{k^*} = \tilde{\pi}^*$) within $k^*$ episodes with probability at least $1 - \delta$.
\end{theorem}

\begin{proof} We give the detailed proof in Appendix~\ref{App:PfLearningProcess}. Let us now sketch the main steps in our analysis. By analyzing the mixing property of the underlying Markov chain of applying $\pi_{\mathrm{rand}}$ to $\tilde{\mathcal{S}}$ and $\pi_0$ to $\mathcal{S} \setminus \tilde{\mathcal{S}}$, we first show that there exists a constant $K_0$ such that for episode $k \geqslant K_0$, we are able to obtain $\Theta(\sqrt{k})$ samples for each $(\bm{Q}, a)$ pair.

We then use a result in \cite{weissman2003inequalities} to compute the number of samples required for each $(\bm{Q}, a)$ pair to ensure $\tilde{\pi}^*_k = \tilde{\pi}^*$, under Assumption~\ref{Asp:ErrDist}. 
Finally, we show that by choosing the exploration probability of $\epsilon_k = \ell / \sqrt{k}$, exploration episodes occur infinitely often, which implies that we can reach the number of required exploration episodes within finite number of episodes.

\end{proof}

Theorem \ref{Thm:LearningProcess} indicates that RL-QN explores (i.e. samples) state-transition functions of each state-action pair $(\bm{Q}, a)$ in $\tilde{\mathcal{M}}$ sufficiently.


\subsection{Gap to Optimum} \label{Sec:Backlog}

Theorem \ref{Thm:LearningProcess} states the sufficient exploration aspect of RL-QN. In reinforcement learning, the trade-off between exploration and exploitation is of significant importance to the algorithm performance. In this section, we show that RL-QN also exploits the learned policy such that the episodic average queue backlog is bounded and can get arbitrarily close to $\rho^*$, the optimal average queue backlog of $\mathcal{M}$, as we increase $U$.

We denote the time step at the end of episode $k$ by $t_k$ and the actual length of episode $k$ by $L_k'$, i.e., $L_k' = t_k - t_{k-1}$ with $t_0 \triangleq 0$. We use $\pi_k^{\mathrm{in}}$ to represent the policy applied to $\mathcal{S}^{\mathrm{in}}$ during episode $k$ and $p^{\tilde{\pi} + \pi_0}(\cdot)$ to denote the stationary distribution of states when applying $\tilde{\pi}$ to states in $\mathcal{S}^{\mathrm{in}}$ and $\pi_0$ to states in $\mathcal{S}^{\mathrm{out}}$.

By Theorem \ref{Thm:LearningProcess}, RL-QN learns $\tilde{\pi}^*$ with high probability. Note that as the exploration probability decays by $1/\sqrt{k}$, the probability of utilizing the learned policy converges to $1$ as the episodes increase. Hence the episodic average queue backlog when $\pi_k^{\mathrm{in}} = \tilde{\pi}^*$ constitutes a large proportion of the overall expected average queue backlog. Therefore, the key step to upper bound the expected average queue backlog is to upper bound the episodic average queue backlog when $\pi_k^{\mathrm{in}} = \tilde{\pi}^*$. We prove that it can be upper bounded with respect to $\tilde{\rho}^*$, the optimal average queue backlog of $\tilde{\mathcal{M}}$, as stated in Lemma~\ref{Lem:EpiExBacklog}. We first define $\mathcal{S}^{\mathrm{in}}_{\mathrm{bd}}$ as the set of states in $\mathcal{S}^{\mathrm{in}}$ that is possible to exit into $\mathcal{S}^{\mathrm{out}}$ at the next time slot, i.e. $\mathcal{S}^{\mathrm{in}}_{\mathrm{bd}} \triangleq \left\{ \bm{Q} \in \mathcal{S}^{\mathrm{in}}: \mathcal{R} \left( \bm{Q} \right) \cap \mathcal{S}^{\mathrm{out}} \neq \emptyset \right\}$.

\begin{lemma} \label{Lem:EpiExBacklog} Under Assumptions~\ref{Asp:KnownPolicy}-\ref{Asp:PolyTrans}, we have
$$
\lim_{k \to \infty} \mathbb{E}_{\tilde{\pi}^* + \pi_0} \left[ \frac{\sum_{t = t_{k-1}+1}^{t_k} \sum_i Q_i(t)}{L_k'} \right] = \tilde{\rho}^* + p^{\tilde{\pi}^* + \pi_0} \left( \mathcal{S}^{\mathrm{in}}_{\mathrm{bd}} \right) \cdot \mathcal{O} \left( U^{1 + \max \{ 2\alpha, \gamma \}} \right) .
$$
\end{lemma}

\begin{proof}
We first define the accumulated regret regarding $\tilde{\rho}^*$ for a given episode $k$ with $\pi_k^{\mathrm{in}} = \tilde{\pi}^*$ as $\sum_{t = t_{k-1}+1}^{t_k} \big(\sum_i Q_i(t) - \tilde{\rho}^*\big)$. We then define $\mathcal{S}^{\mathrm{in}}_{\mathrm{in}}  \triangleq \mathcal{S}^{\mathrm{in}} \setminus \mathcal{S}^{\mathrm{in}}_{\mathrm{bd}}$ and decompose the expected accumulated regret into three parts according to state position: (i) $\bm{Q} \in \mathcal{S}^{\mathrm{in}}_{\mathrm{in}}$, (ii) $\bm{Q} \in \mathcal{S}^{\mathrm{in}}_{\mathrm{bd}}$ and (iii) $\bm{Q} \in \mathcal{S}^{\mathrm{out}}$. 

For the first part (i), we use Bellman equation analysis to show that, every time $\bm{Q}$ enters $\mathcal{S}^{\mathrm{in}}_{\mathrm{in}}$ from $\mathcal{S}^{\mathrm{in}}_{\mathrm{bd}}$ and returns back to $\mathcal{S}^{\mathrm{in}}_{\mathrm{bd}}$, the expected accumulated regret is upper bounded by the span of the solutions to the Bellman equations. We then use Proposition 5.5.1 in \cite{bertsekas2017dynamic} to obtain a polynomial upper bound for the span of 
the solution to the Bellman equation under Assumption~\ref{Asp:PolyTrans}.

The second part (ii) is trivially upper bounded by $DU$. 

For the third part (iii), we  prove that the time it takes to return back $\mathcal{S}^{\mathrm{in}}$ is polynomial in $U$ under  Assumption~\ref{Asp:KnownPolicy} (using techniques as the proof of Theorem 1.1 in Chapter 5 of \cite{bremaud1999lyapunov}). Therefore, by Theorem 6.3.4 in \cite{hou2012homogeneous}, every time $\bm{Q}$ enters $\mathcal{S}^{\mathrm{out}}$ from $\mathcal{S}^{\mathrm{in}}_{\mathrm{bd}}$ and returns back to $\mathcal{S}^{\mathrm{in}}_{\mathrm{bd}}$, the incurred expected accumulated regret is upper bounded by a polynomial function of $U$.

The detailed proof is given in Appendix~\ref{App:PfEpiExBacklog}.
\end{proof}

We then proceed to upper bound $p^{\tilde{\pi}^* + \pi_0} \left( \mathcal{S}^{\mathrm{in}}_{\mathrm{bd}} \right),$ as stated in the following lemma.

\begin{lemma} \label{Lem:pBd} Under Assumption~\ref{Asp:DriftOptimalPolicy}, we have
$$
p^{\tilde{\pi}^* + \pi_0} \left( \mathcal{S}^{\mathrm{in}}_{\mathrm{bd}} \right) = \mathcal{O} \left( \exp \left( - U \right) \right) .
$$
\end{lemma}

\begin{proof}
We show that under Assumption~\ref{Asp:DriftOptimalPolicy}, we can construct a linear Lyapunov function with a negative drift for states with large $\bm{Q}_{\max}$. By applying similar techniques as Lemma 1 in \cite{bertsimas1998geometric}, we establish an upper bound for the tail probability of Lyapunov values, which decays exponentially.

The detailed proof is given in Appendix~\ref{App:pBd}.
\end{proof}

With Theorem~\ref{Thm:LearningProcess}, Lemma \ref{Lem:EpiExBacklog} and Lemma \ref{Lem:pBd}, we establish the following main result of this paper.

\begin{theorem} \label{Thm:Backlog} Suppose Assumptions~\ref{Asp:KnownPolicy}-\ref{Asp:ErrDist} hold. When applying \emph{RL-QN} to $\mathcal{M}$, the asympototic episodic average queue backlog is upper bounded as follows:
\begin{align}
 \lim_{k \to \infty} \mathbb{E} \left[ \frac{\sum_{t = t_{k-1}+1}^{t_k} \sum_i Q_i(t)}{L_k'} \right] \leq \rho^* + \mathcal{O} \left( \frac{U^{1 + \max \{ 2\alpha, \gamma \}}}{\exp \left( U \right) } \right) .  \label{eq:episodic_bound} 
\end{align}
\end{theorem}

\begin{proof}
With Lemma~\ref{Lem:EpiExBacklog} and Lemma~\ref{Lem:pBd}, we show that when $k \to \infty$, for each episode with $\pi_k^{\mathrm{in}} = \tilde{\pi}^*$, the expected average queue backlog is upper bounded by $\tilde{\rho}^* + \mathcal{O} (U^{1 + \max \{ 2\alpha, \gamma \}} / \exp (U))$. We next show that the expected average queue backlog contributed by episodes where $\pi_k^{\mathrm{in}} \neq \tilde{\pi}^*$ becomes negligible as $k \to \infty$. By applying similar techniques as the proof of Lemma~\ref{Lem:EpiExBacklog} and Lemma~\ref{Lem:pBd}, we then obtain that $\tilde{\rho}^* = \rho^* + \mathcal{O} ( U^{1 + \gamma } / \exp \left( U \right) )$.

The detailed proof is given in Appendix~\ref{App:PfBacklog}.
\end{proof}

Theorem~\ref{Thm:Backlog} provides an upper bound on the performance guarantee of RL-QN with respect to the threshold parameter $U$: by increasing $U$, the long-term episodic average queue backlog approaches $\rho^*$ exponentially fast. Recall that the episodic length $L'_k$ increases to $\infty$ as $k \rightarrow \infty$. We conjecture that the same upper bound hold for the overall average queue backlog regarding the time horizon $T,$ i.e., 
\begin{align}
\lim_{T \to \infty} \frac{\mathbb{E} \left[\sum_{t = 1}^{T} \sum_i Q_i(t) \right]}{T} \leq \rho^* + \mathcal{O} \left( \frac{U^{1 + \max \{ 2\alpha, \gamma \}}}{\exp \left( U \right) } \right) . \label{eq:conjecture}
\end{align}
We note that the result of Theorem~\ref{Thm:Backlog} does not imply (\ref{eq:conjecture}) directly. A rigorous proof of the conjecture~(\ref{eq:conjecture}) seems difficult with current techniques. We leave as an interesting future direction to investigate if (\ref{eq:conjecture}) holds.


\subsection{Complexity Analysis} \label{Sec:Complexity}

Here we present the complexity analysis. Our algorithm requires solving an estimated MDP for each exploitation episode. Since episode $k$ has length of at least $L \cdot \sqrt{k}$, and
$$
  \sum_{k=1}^{K} \sqrt{k} \geqslant \frac{2}{3} K^{\frac{3}{2}} ,
$$
we have that given a time horrzon $T$, the number of episodes is upper bounded as
$$
  K_T \leqslant \bigg( \frac{3T}{2L} \bigg)^{\frac{2}{3}} .
$$
It has been shown that to solve a general MDP, the time complexity is at least polynomial in the size of the state space~\cite{littman2013complexity}. The concrete time complexity depends on the solution method and its parameters. In our setting, the state space $\mathcal{S}^{in}$ has a size of $\Theta(U^N)$. Therefore, the time complexity of our algorithm is
$$
\mathcal{O} \Big( \big( T/L \big)^{\frac{2}{3}} \cdot \text{poly} \big( U^N \big) \Big) .
$$

Therefore, for a given system with fixed number of nodes, the time complexity grows at most polynomially in the threshold $U$. When applying our algorithm to problems of larger scales (larger $N$), the complexity suffers from ``curse of dimensionality''. However, our algorithm remains computationally feasible in practice for the following reasons. First, although the ``curse of dimensionality'' persists, our algorithm substantially simplifies the original MDP with unbounded state space to an MDP with finite state space. For instance, the system in Section 5.1 cannot be optimized by traditional methods but is optimized efficiently under our algorithm. Second, we only require solving the estimated MDP sparsely, i.e., no more than $(3T/2L)^{2/3}$ times. Choosing a relatively large $L$ can greatly reduce the time complexity. Third, solving the MDP is independent of other steps of our algorithm, which allows us to employ appropriate methods to solve MDPs according to different application scenarios, computational capacities and performance requirements. For instance, since we aim at optimizing the average cost of the MDP, undiscounted value iteration/policy iteration methods should be applied. However, we apply discounted methods in our numerical experiments and obtain the same solution with significantly less computation time. Also, we may utilize the special structure of the MDP \cite{guestrin2003efficient,koller2013policy} or apply approximation methods \cite{chen1999value,chen2018deep,veatch2010approximate} to alleviate the computational cost. Our numerical results in Section \ref{Sec:SA_2} further validates the conclusion and shows that our algorithm is computationally feasible in practice.


\section{Numerical Experiments} \label{Sec:Sim}

In this section, we evaluate the performance of RL-QN for three classes of queueing systems:  server allocation, routing, switching.


\subsection{Server Allocation (Two Nodes)} \label{Sec:SA_2}

We consider a dynamic server allocation problem: exogenous packets arrive to two nodes according to Bernoulli process with rate $\lambda_1$ and $\lambda_2$ respectively. At each time slot, the central server selects one node to serve. The head of line job in the selected queue $i$ completes the required service and leaves the system with probability $p_i$. The system model and parameters are illustrated in Figure~\ref{Fig:SA_2_Model}.

\begin{figure}[H]
\centering
\begin{minipage}{.5\textwidth}
  \centerline{\includegraphics[width=0.9\linewidth]{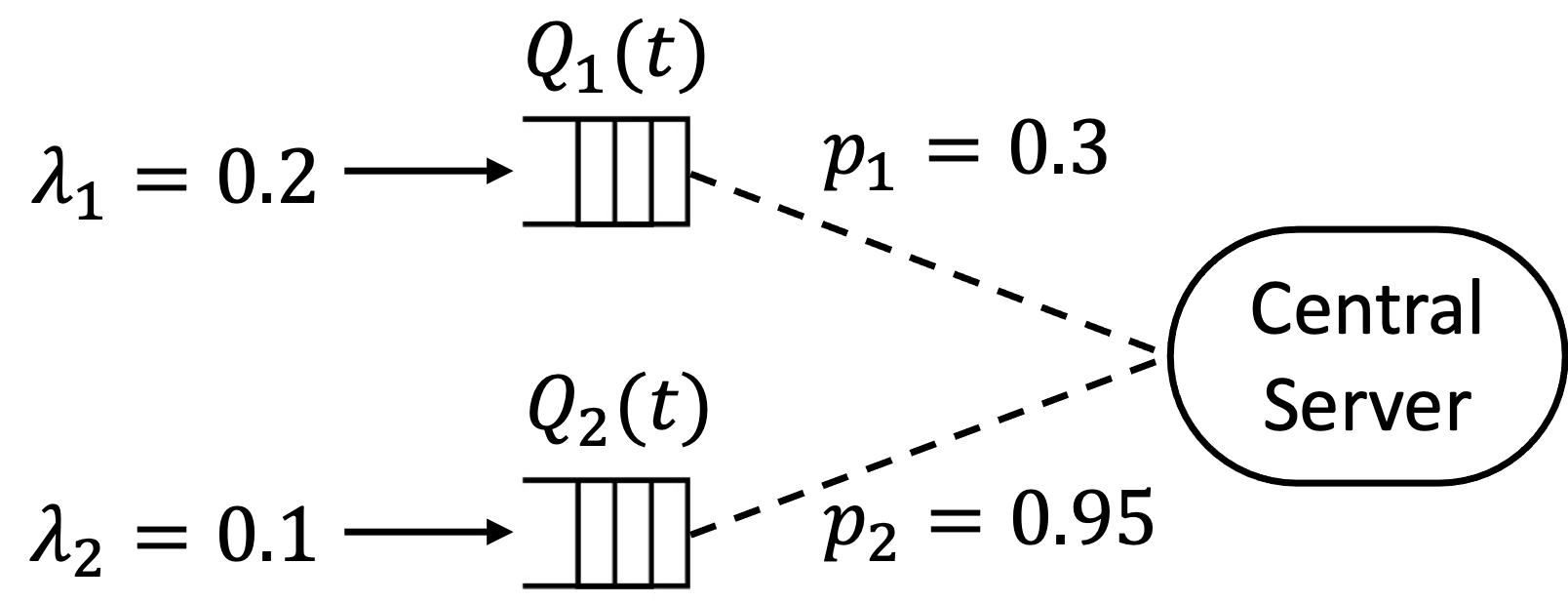}}
  \caption{System model of the dynamic server allocation model with two nodes.}
  \label{Fig:SA_2_Model}
\end{minipage}%
\begin{minipage}{.5\textwidth}
  \centerline{\includegraphics[width=0.9\linewidth]{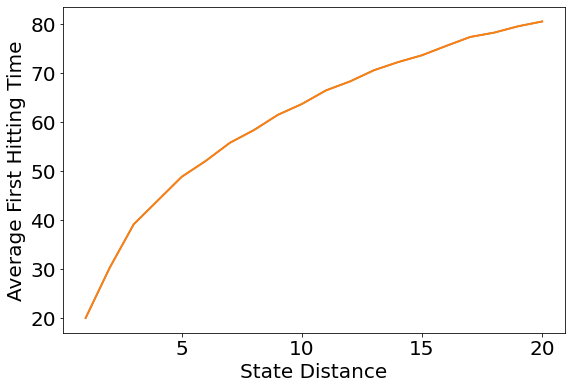}}
  \caption{Simulation results for the dynamic server allocation model of two nodes.}
  \label{Fig:SA_2_HT}
\end{minipage}
\end{figure}

According to \cite{tassiulas1993dynamic}, whenever the condition
$$
\frac{\lambda_1}{p_1} + \frac{\lambda_2}{p_2} < 1 ,
$$
is satisfied, one stabilizing policy is to always serve the node with the longest queue (LQ). Therefore, we can use LQ policy as $\pi_0$. To evaulate whether this problem satisfies Assumption \ref{Asp:PolyTrans}, we apply $\pi_0$ to the truncated state space $\tilde{\mathcal{S}}$ with $U = 10$ and simulate the system to collect the hitting times, as shown in Figure \ref{Fig:SA_2_HT}. We can see that as the state distance grows, the average first hitting time grows sublinearly, which is obviously upper bounded by linear growth and thus indicates that Assumption \ref{Asp:PolyTrans} is satisfied. Moreover, our numerical experiments show that our algorithm can optimize the average queue backlog for various settings even if we cannot provably verify Assumption \ref{Asp:PolyTrans}. Therefore, the applicability of our algorithm may not be sensitive to Assumption \ref{Asp:PolyTrans}, which helps to show that our algorithm is practical.

\begin{figure}[htbp]
\centerline{\includegraphics[width=0.5\linewidth]{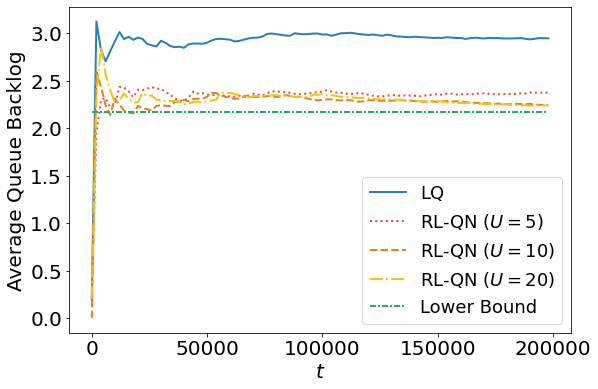}}
\caption{Average queue backlog under different policies for the server allocation model with two nodes.}
\label{Fig:SA_2_Result}
\end{figure}

We simulate RL-QN for $U=5$, $U=10$ and $U=20$. The results are shown in Figure~\ref{Fig:SA_2_Result}. To obtain a lower bound on the average queue backlog, we assume that the system dynamics, i.e., arrival rates and success rates, are known. We then optimize the MDP and obtain its average queue backlog. This average queue backlog is guaranteed to be a lower bound since, in practice, the system dynamics are unknown and errors that occur during the learning process can downgrade the performance. From Figure~\ref{Fig:SA_2_Result}, we see that the longest queue (LQ) policy stabilizes the queue backlog, yet its average queue backlog is far from the lower bound. All of our RL-QN methods outperform $\pi_0$. When $U = 5$, the average queue backlog converges to $2.38$, while for $U = 10$ and $U = 20$, the average queue backlog becomes $2.24$. This indicates that RL-QN achieves better performance with a larger threshold parameter $U$, as implied by Theorem~\ref{Thm:Backlog}. Moreover, since the cases with $U = 10$ and $U = 20$ achieve similar performance, it is very likely that in practice a small $U$ achieves satisfactory performance, which makes our algorithm more computationally feasible.

\begin{figure}[H]
\centering
\begin{minipage}{.48\textwidth}
  \centerline{\includegraphics[width=0.95\linewidth]{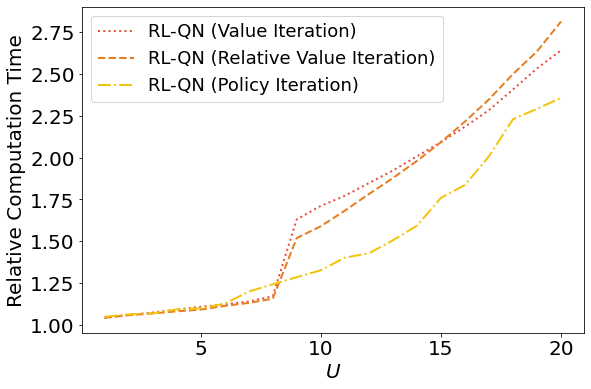}}
  \caption{The relationship between $U$ and the relative \\ computation time.}
  \label{Fig:SA_2_TU}
\end{minipage}%
\begin{minipage}{.48\textwidth}
  \centerline{\includegraphics[width=0.95\linewidth]{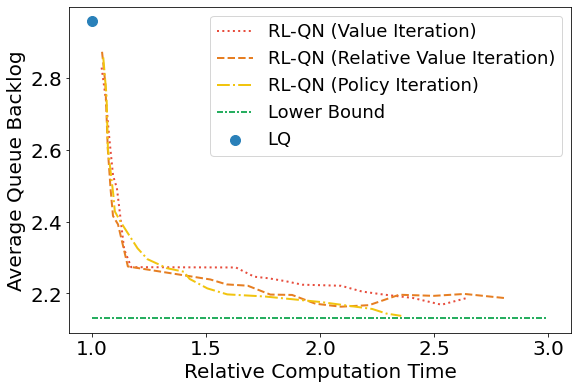}}
  \caption{The relationship between the relative computation time \\ and the average queue backlog.}
  \label{Fig:SA_2_QT}
\end{minipage}
\end{figure}

We then study the time complexity of our algorithm. We apply three different MDP solvers in our algorithm: value iteration, relative value iteration and policy iteration, all with discount factor $0.99$. Figure \ref{Fig:SA_2_TU} shows the relationship between the relative computation time (i.e., the actual computation time of RL-QN devided by the actual computation time of LQ) and the average queue backlog. Even when $U = 20$, the computation time is less than three times of the LQ policy, which shows that increasing $U$ has relatively small impact on the computation time. From Figure \ref{Fig:SA_2_QT}, we see that the average queue backlog drops significantly even with a relatively small $U$. The average queue backlog is reduced by $23.1\%$ when $U = 8$, while the computation time is only $1.17$ times of the traditional non-ML LQ policy. Therefore, our algorithm is computationally efficient in practice.


\subsection{Server Allocation (Ten Nodes)} \label{Sec:SA_10}

We then consider a dynamic server allocation problem with greater scale: exogenous packets arrive to ten nodes according to Bernoulli process with rate $\lambda_i, \ i = 1, \cdots, 10$ respectively. At each time slot, the central server selects one node to serve. The head of line job in the selected queue $i$ completes the required service and leaves the system with probability $p_i$. The system model and parameters are illustrated in Figure~\ref{Fig:SA_10_Model}.

\begin{figure}[h]
\centerline{\includegraphics[width=0.9\linewidth]{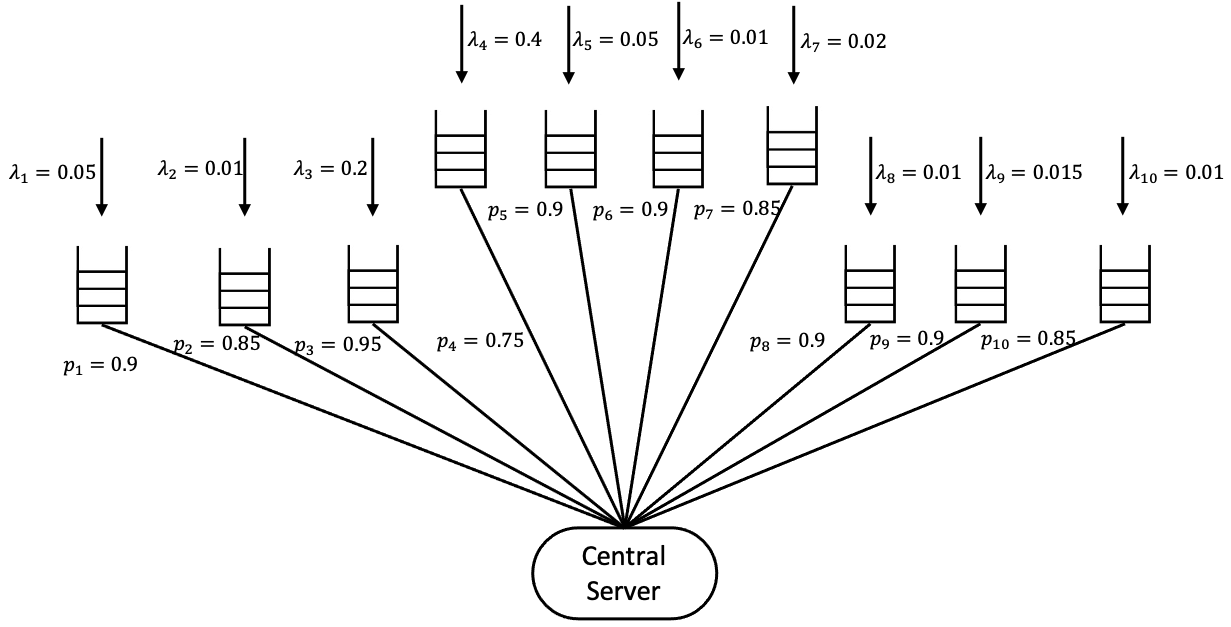}}
\caption{System model of the dynamic server allocation model with ten nodes.}
\label{Fig:SA_10_Model}
\end{figure}

As discussed in Section \ref{Sec:Outline}, if we directly apply classical methods to solve the estimated MDP, the computational complexity is at least $\Theta(U^N)$, which grows exponentially with the number of nodes $N$. Therefore, for the ten-node dynamic server allocation model, we prefer to utilize the special structures of the MDP or apply approximation methods to reduce the computational cost.

It has been shown in \cite{buyukkoc1985cmu} that when all nodes are always connected to the central server, the policy to minimize the average job delay is to serve the node with the largest service rate $p_i$ among all non-empty nodes. Therefore, for each exploitation episode, instead of solving the estimated MDP $\tilde{\mathcal{M}}_k$ with classical methods, we directly obtain the optimal solution: serving the node with the largest estimated $p_i$ among all non-empty nodes. We emphasize that our algorithm does not depend on how the estimated MDPs are solved. The above computational trick utilizes the special structure of the server allocation model to speed up the computation, and other techniques for efficiently solving the MDPs can also potentially be employed.

\begin{figure}[htbp]
\centerline{\includegraphics[width=0.5\linewidth]{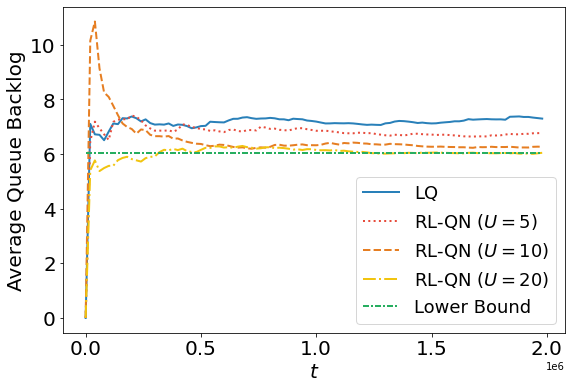}}
\caption{Simulation results for the server allocation model with ten nodes.}
\label{Fig:SA_10_Result}
\end{figure}

The results are shown in Figure~\ref{Fig:SA_10_Result}. From the figure, we observe that RL-QN outperforms LQ. The LQ reaches an average queue backlog of $7.6$, while RL-QN reaches $7$ when $U=5$. When $U$ is beyond $10$, RL-QN reaches an average queue backlog of $6$, which is close to the optimal result.


\subsection{Routing} \label{Sec:Routing_2}

We consider a simple routing problem: exogenous packets arrive at the source node $s$ according to Bernoulli process with rate $\lambda = 0.85$. Node $1$ and node $2$ are two intermediate nodes and can serve at most one packet during each time slot, with probability $p_1$ and $p_2$ respectively. Node $d$ is the destination node. At each time slot, node $s$ has to choose between routes $s \to 1 \to d$ and $s \to 2 \to d$ to transit new exogenous packets. Specifically, the system model and parameters are shown in Figure~\ref{Fig:Routing_2_Model}.

The parameters $(p_1,p_2)$ are queue-dependent here:
$$
(p_1, p_2) = 
\begin{cases}
(0.9, 0.1), & Q_2(t) \leqslant 5 \\
(0.1, 0.9), & Q_2(t) > 5
\end{cases}
.
$$

For each $\lambda < 0.9$, an intuitive stabilizing policy is to always use the fixed path $s \to 1 \to d$, while never choose $s \to 2 \to d$. Therefore, we can use the policy that always routes through $s \to 1 \to d$ as $\pi_0$. However, it is possible that we could split the external arrivals into the two routes to fully utilize the service capacities of both node $1$ and $2$, and achieve better performance.

We simulate RL-QN for $U=10$. The results are plotted in Figure~\ref{Fig:Routing_2_Result}, which shows that RL-QN outperforms the fixed path stabilizing policy and converges to the optimum quickly.

\begin{figure}[H]
\centering
\begin{minipage}{.4\textwidth}
  \centerline{\includegraphics[width=0.9\linewidth]{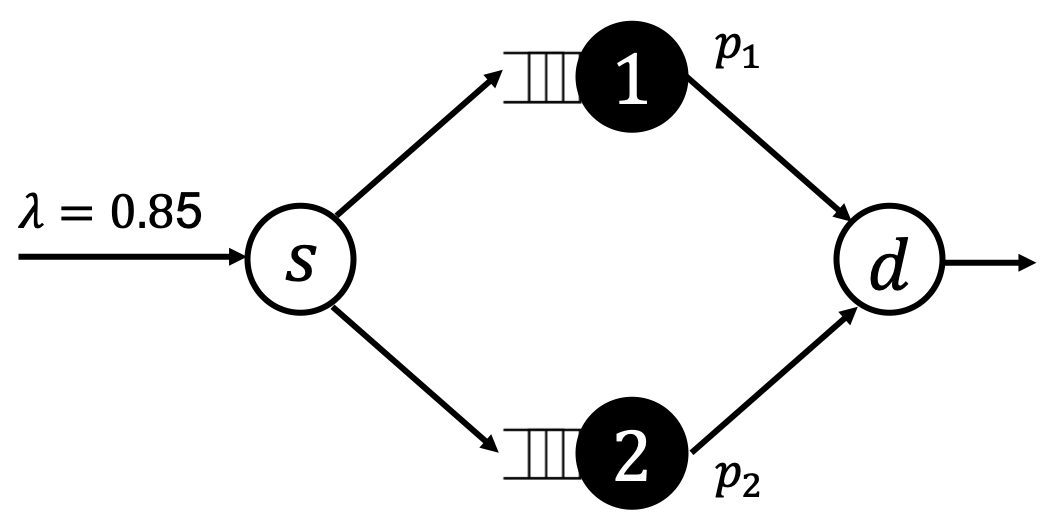}}
  \caption{System model of the routing network with two nodes.}
  \label{Fig:Routing_2_Model}
\end{minipage}%
\begin{minipage}{.5\textwidth}
  \centerline{\includegraphics[width=0.9\linewidth]{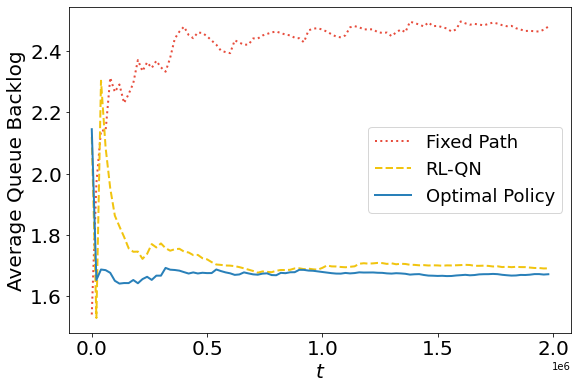}}
  \caption{Simulation results of the routing network with two nodes.}
  \label{Fig:Routing_2_Result}
\end{minipage}
\end{figure}


\subsection{Switching} \label{Sec:Switching}

We consider a $2 \times 2$ input-queued switch as illustrated in Figure~\ref{Fig:Switching_Model}. Data packets arriving at input $i$ destined for output $j$ are stored at input port $i,$ in queue $Q_{i,j}$, thus there are $4$ queues in total. 
We consider the case where the new data packets are arriving
at queue $(i,j)$ at rate $\lambda_{i,j}$ for $1\leq i,j \leq 2$, according to a Bernoulli
process. That is, for each time slot, the number of packets arriving at queue $Q_{i,j}$ is
a Bernoulli random variable with mean $\lambda_{i,j}$.
The server then selects a matching between the inputs and outputs to transmit packets. If input $i$ is connected with output $j$, then a buffered packet is removed from the input queue $Q_{i,j}$ and sent to output $j$.
\begin{figure}[htbp]
\centerline{\includegraphics[width=0.5\linewidth]{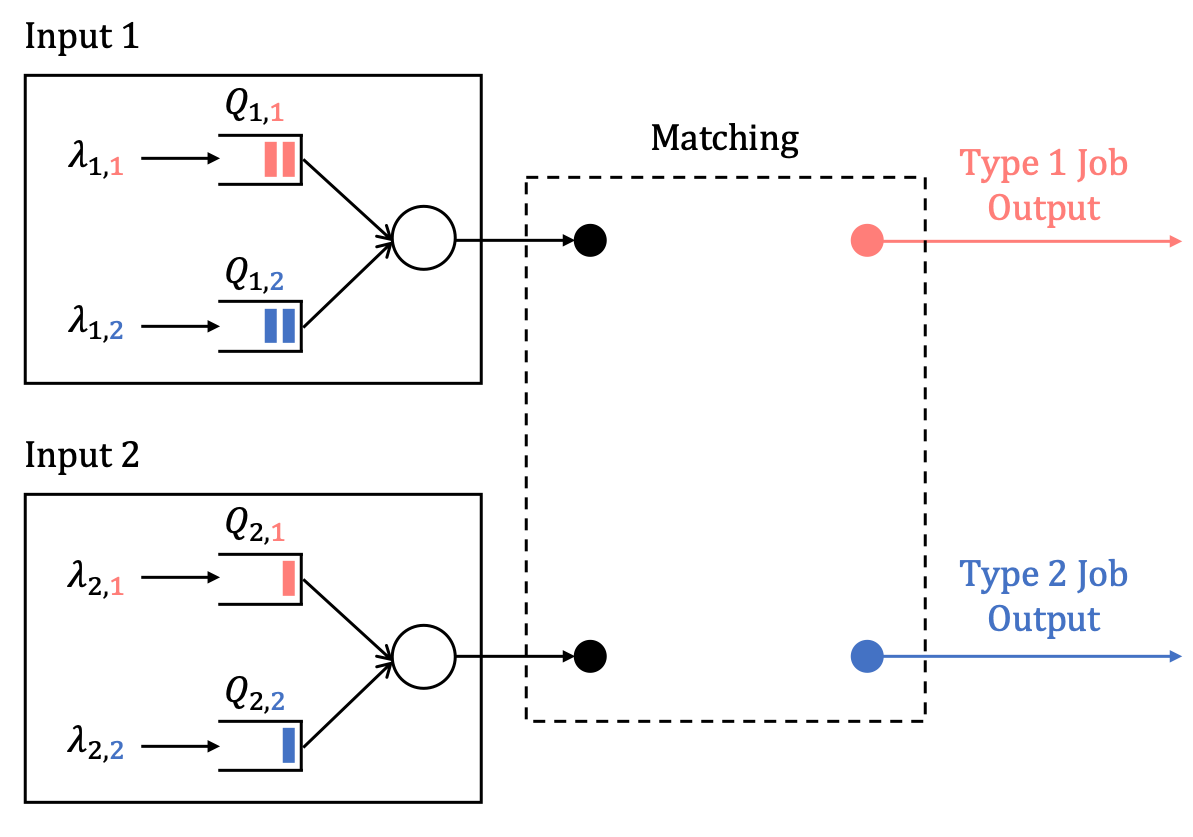}}
\caption{System model of a $2 \times 2$ input-queued cell-switch.}
\label{Fig:Switching_Model}
\end{figure}

According to \cite{mckeown1999achieving}, whenever the condition that $\sum_{i} \lambda_{i, j} < 1, \sum_{j} \lambda_{i, j} < 1$ is satisfied, the Maximum Matching algorithm which selects the matching that maximizes the total queue backlog of the connected channels is stabilizing, hence we use it as $\pi_0$.

\begin{figure}[H]
    \centering
    \begin{subfigure}[b]{0.45\textwidth} 
        \centering
        \includegraphics[width=\linewidth]{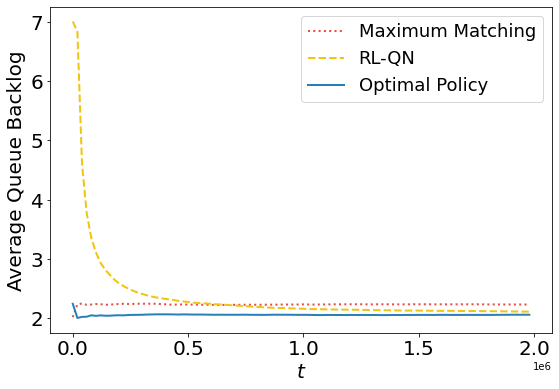}
        \caption{$\lambda_{1, 1} = \lambda_{1, 2} = 0.4$, $\lambda_{2, 1} = 0.2$, $\lambda_{2, 2} = 0.1$}
        \label{Fig:Switching_asy_Result}
    \end{subfigure}
    \hfill
    \begin{subfigure}[b]{0.45\textwidth} 
        \centering
        \includegraphics[width=\linewidth]{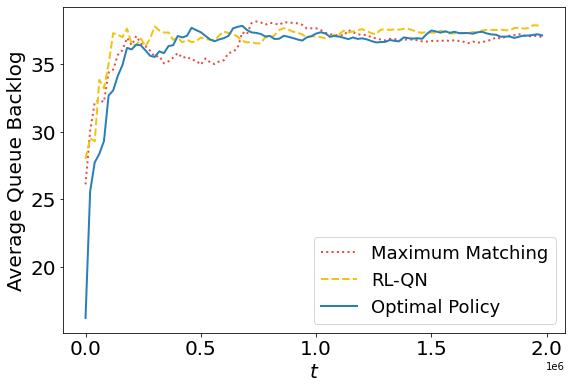}
        \caption{$\lambda_{1, 1} = \lambda_{1, 2} = \lambda_{2, 1} = \lambda_{2, 2} = 0.49$}
        \label{Fig:Switching_heavy_Result}
    \end{subfigure}
    \caption{Simulation results for the input-queued switch model.}
    \label{Fig:Switching_Result}
\end{figure}

We implement the simulation for $U=5$ under two different settings of arrival traffic. The results are shown in Figure~\ref{Fig:Switching_Result}. When $\lambda_{1, 1} = \lambda_{1, 2} = 0.4$, $\lambda_{2, 1} = 0.2$, $\lambda_{2, 2} = 0.1$, we can see that our algorithm outperforms the stabilizing policy $\pi_0$ and converges to $\tilde{\pi}^* + \pi_0$. When $\lambda_{1, 1} = \lambda_{1, 2} = \lambda_{2, 1} = \lambda_{2, 2} = 0.49$, the system is under heavy traffic. In \cite{maguluri2016heavy}, it was shown that the Maximum Matching policy $\pi_0$  is close to the optimal policy that minimizes the average queue backlog. Our simulation result (Figure~\ref{Fig:Switching_heavy_Result}) indicates that our algorithm achieves comparable performance as the near-optimal policy $\pi_0.$


\section{Conclusion} \label{Sec:Conclusion}

In this work, we apply a model-based reinforcement learning framework to general queueing networks with unbounded state spaces. We propose the RL-QN algorithm, which applies an $\epsilon$-greedy exploration scheme. By leveraging Lyapunov analysis, we show that the average queue backlog of the proposed approach can get arbitrarily close to the optimal average queue backlog under the optimal policy. The proposed RL-QN algorithm requires the knowledge of a stable policy. An interesting future direction is to investigate this problem when such information is not available.


\appendix
\appendixpage

\section{Proof of Theorem \ref{Thm:LearningProcess}} \label{App:PfLearningProcess}

In this section, we prove Theorem~\ref{Thm:LearningProcess}. As outlined in Section~\ref{Sec:Conv2Policy}, our proof consists of three steps, which are presented in the subsections to follow.

\subsection{Sufficient Exploration for $\tilde{\mathcal{S}}\times{\mathcal{A}}$}

For each $(\bm{Q}, a)\in \tilde{\mathcal{S}}\times{\mathcal{A}}$, we denote by $N_k(\bm{Q}, a)$ the number of times that $(\bm{Q}, a)$ is encountered during episode $k$. For simplicity, we use $\tilde{\pi}_k$ to denote the policy applied to $\tilde{\mathcal{S}}$ during episode $k$ and $p^{\mathrm{rand}}(\cdot)$ to denote the stationary distribution of states when applying $\pi_{\mathrm{rand}}$ to states in $\tilde{\mathcal{S}}$ and $\pi_0$ to states in $\mathcal{S} \setminus \tilde{\mathcal{S}}$. The following lemma illustrates that after a certain number of episodes, $\pi_{\mathrm{rand}}$ samples each $(\bm{Q}, a) \in \tilde{\mathcal{S}} \times \mathcal{A}$ sufficiently with a relatively large probability (e.g., greater than $1/2$). 

\begin{lemma} \label{Lem:Mix}
Under the setting of Theorem~\ref{Thm:LearningProcess} and Algorithm \ref{Alg:RL-QN}, there exists $K_0 > 0$ such that for each $k \geqslant K_0$,
$$
\Pr \left\{ \frac{N^{\mathrm{rand}}_k(\bm{Q}, a)}{L_k} \leqslant \frac{p^{\mathrm{rand}} \left( \bm{Q} \right)}{2\abs{\mathcal{A}}},\forall \bm{Q} \in \tilde{\mathcal{S}},\forall a \in \mathcal{A} \right\} \geqslant \frac{1}{2}.
$$
\end{lemma}

\begin{proof}

Here we only consider the episodes that $\tilde{\pi}_k = \pi_{\mathrm{rand}}$. (That is, in this proof episode $k$ is understood as the episode in which the policy $\pi_{\mathrm{rand}}$ is executed for the $k$-th time.) 
Recall that under the policy that applies $\pi_{\mathrm{rand}}$ to states in $\tilde{\mathcal{S}}$ and $\pi_0$ to states in $\mathcal{S} \setminus \tilde{\mathcal{S}}$, the corresponding Markov chain is positive recurrent.
For each $\bm{Q}\in \mathcal{S},$ define $N^{\mathrm{rand}}_k \left( \bm{Q} \right)$ as the number of times that $\bm{Q}$ is visited during episode $k$. For an irreducible positive recurrent Markov chain a on countable state space, we have the mixing property that for every $\bm{Q} \in \mathcal{S},$
\begin{equation} \label{Eqn:ASQ2}
\lim_{k \to \infty} \frac{N^{\mathrm{rand}}_k \left( \bm{Q} \right)}{L'_k} = p^{\mathrm{rand}} \left( \bm{Q} \right) \qquad w.p. 1 .
\end{equation}

Note that under $\pi_{\mathrm{rand}}+\pi_0$, for every $\bm{Q} \in \tilde{\mathcal{S}}$, we take each $a \in \mathcal{A}$ with equal probability $1 / \abs{\mathcal{A}}$. According to strong law of large number, we have
\begin{equation} \label{Eqn:ASQA2}
\lim_{k \to \infty} \frac{N^{\mathrm{rand}}_k(\bm{Q}, a)}{N^{\mathrm{rand}}_k \left( \bm{Q} \right)} = \frac{1}{\abs{\mathcal{A}}} \qquad w.p. 1 .
\end{equation}
for every $\bm{Q} \in \tilde{\mathcal{S}}$ and $a \in \mathcal{A}$.

Since both \eqref{Eqn:ASQ2} and \eqref{Eqn:ASQA2} converge to constants, the multiplication rule for limit holds. That is, for every $\bm{Q} \in \tilde{\mathcal{S}}$ and $a \in \mathcal{A}$, we have 
\begin{align}
\lim_{k \to \infty} \frac{N^{\mathrm{rand}}_k(\bm{Q}, a)}{L'_k} = & \lim_{k \to \infty} \frac{N^{\mathrm{rand}}_k \left( \bm{Q} \right)}{L'_k} \cdot \frac{N^{\mathrm{rand}}_k(\bm{Q}, a)}{N^{\mathrm{rand}}_k \left( \bm{Q} \right)} \nonumber \\
= & \lim_{k \to \infty} \frac{N^{\mathrm{rand}}_k \left( \bm{Q} \right)}{L'_k} \cdot \lim_{k \to \infty} \frac{N^{\mathrm{rand}}_k(\bm{Q}, a)}{N^{\mathrm{rand}}_k \left( \bm{Q} \right)} \nonumber \\
= & \frac{p^{\mathrm{rand}} \left( \bm{Q} \right)}{\abs{\mathcal{A}}} \qquad w.p. 1. \nonumber
\end{align}

Note that almost sure convergence implies convergence in probability. Hence for each $(\bm{Q},a) \in \tilde{\mathcal{S}} \times \mathcal{A}$ and $\epsilon, \xi > 0$, there exists a finite constant such that for each $k$ larger than the constant,
$$
\Pr \left\{ \abs{\frac{N^{\mathrm{rand}}_k(\bm{Q}, a)}{L'_k} - \frac{p^{\mathrm{rand}} \left( \bm{Q} \right)}{\abs{\mathcal{A}}}} \geqslant \epsilon \right\} \leqslant \xi.
$$

Since $L_k'\geq L_k$, we have 
\begin{align}
\Pr \left\{ \frac{N^{\mathrm{rand}}_k(\bm{Q}, a)}{L_k} \leqslant \frac{p^{\mathrm{rand}} \left( \bm{Q} \right)}{2\abs{\mathcal{A}}} \right\} \leqslant & \Pr \left\{ \frac{N^{\mathrm{rand}}_k(\bm{Q}, a)}{L'_k} \leqslant \frac{p^{\mathrm{rand}} \left( \bm{Q} \right)}{2\abs{\mathcal{A}}} \right\} \nonumber \\
\leqslant & \Pr \left\{ \abs{\frac{N^{\mathrm{rand}}_k(\bm{Q}, a)}{L'_k} - \frac{p^{\mathrm{rand}} \left( \bm{Q} \right)}{\abs{\mathcal{A}}}} \geqslant \frac{p^{\mathrm{rand}} \left( \bm{Q} \right)}{2\abs{\mathcal{A}}} \right\} . \nonumber
\end{align}
By setting $\epsilon = p^{\mathrm{rand}} \left( \bm{Q} \right) / (2 \abs{\mathcal{A}})$, $\xi = 1/\left(2 \abs{\tilde{\mathcal{S}}} \abs{\mathcal{A}} \right)$ and taking a union bound over $\tilde{\mathcal{S}}$ and $\mathcal{A}$, we have that there exists some constant $K_0 < \infty$ such that for each $k \geqslant K_0$,
$$
\Pr \left\{ \frac{N^{\mathrm{rand}}_k(\bm{Q}, a)}{L_k} \leqslant \frac{p^{\mathrm{rand}} \left( \bm{Q} \right)}{2\abs{\mathcal{A}}},\forall \bm{Q} \in \tilde{\mathcal{S}} \mbox{ and } \forall a \in \mathcal{A} \right\} \leqslant \frac{1}{2 \abs{\tilde{\mathcal{S}}} \abs{\mathcal{A}}} \cdot \abs{\tilde{\mathcal{S}}} \cdot \abs{\mathcal{A}} = \frac{1}{2}.
$$
We complete the proof.

\end{proof}

\subsection{Sample Requirement for learning $\tilde{\pi}^*$}

Define $R \triangleq \max_{\bm{Q} \in \tilde{\mathcal{S}}, a \in \mathcal{A}} \abs{\mathcal{R} \left( \bm{Q}, a \right)}$. We  have the following lemma on the number of samples required for each $(\bm{Q}, a)\in \tilde{\mathcal{S}}, a \in \mathcal{A}$ to estimate the model of the auxiliary system $\tilde{\mathcal{M}}$ sufficiently accurate.

\begin{lemma} \label{Lem:SampleNum}
Given $\delta > 0$, if for each $(\bm{Q},a)\in \tilde{\mathcal{S}}\times \mathcal{A}$, the number of samples  $N(\bm{Q},a)$ satisfies
$$
N(\bm{Q}, a) \geqslant \frac{2}{(\Delta p)^2} \cdot \log \frac{2^{R+1} (U+1)^D \abs{\mathcal{A}}}{\delta} ,
$$
then with probability at least $1 - \frac{\delta}{2}$, the optimal solution of the estimated truncated MDP is exactly $\tilde{\pi}^*$.
\end{lemma}

\begin{proof}
Our proof is based on the following lemma from \cite{weissman2003inequalities}.

\begin{lemma} [Theorem 2.1 in \cite{weissman2003inequalities}] \label{Lem:L1Ineq}
For a probability distribution $p$ over $n_1$ distinct events, we obtain the empirical distribution $\hat{p}$ based on $n_2$ samples from $p$. Then the $L^1-$ deviation between $p$ and $\hat{p}$ is upper bounded as
$$
\Pr \left\{ \norm{ p(\cdot) - \hat{p}(\cdot) }_1 \geqslant \epsilon \right\} \leqslant \left( 2^{n_1}-2 \right) \exp \left( -\frac{n_2 \epsilon^2}{2} \right) .
$$
\end{lemma}

In our case, for a given $(\bm{Q}, a)\in \tilde{\mathcal{S}}\times\mathcal{A}$, the true distribution over next state $\mathcal{R}(\bm{Q})$ is given by $\tilde{p} \left( \cdot \mid \bm{Q}, a \right)$. We denote the empirical distribution by $\hat{p} \left( \cdot \mid \bm{Q}, a \right)$. Note that $ \abs{\mathcal{R} \left( \bm{Q} \right)} \leqslant R$. 
We set 
$$
n_2 \geqslant \frac{2}{(\Delta p)^2} \cdot \log \frac{ 2^{R+1} (U+1)^D \abs{\mathcal{A}}}{\delta}.
$$
By Lemma~\ref{Lem:L1Ineq}, we have that for each $(\bm{Q}, a)\in \tilde{\mathcal{S}}\times\mathcal{A}$, 
\begin{align}
& \Pr \left\{ \sum_{\bm{Q}' \in \mathcal{R} \left( \bm{Q} \right)} \abs{\tilde{p}( \bm{Q}' \mid \bm{Q}, a) - \hat{p}( \bm{Q}' \mid \bm{Q}, a)} \geqslant \Delta p \right\} \nonumber \\
\leqslant & (2^R-2) \cdot \exp \left( - \frac{(\Delta p)^2}{2} \cdot \frac{2}{(\Delta p)^2} \cdot \log \frac{2^{R+1} (U+1)^D \abs{\mathcal{A}}}{\delta} \right) \nonumber \\
\leqslant & \frac{\delta}{2(U+1)^D \abs{\mathcal{A}}}. \nonumber
\end{align}

Taking a union bound over each $\bm{Q} \in \tilde{\mathcal{S}}$ and $a \in \mathcal{A}$, we obtain
\begin{align}
& \Pr \left\{ \text{there exists } (\bm{Q},a) \in \tilde{\mathcal{S}} \times \mathcal{A} \text{ such that } \norm{ \tilde{p} \left( \cdot \mid \bm{Q}, a \right) - \hat{p} \left( \cdot \mid \bm{Q}, a \right) }_1 \geqslant \Delta p \right\} \nonumber \\
\leqslant & \frac{\delta}{2(U+1)^D \abs{\mathcal{A}}} \cdot (U+1)^D \cdot \abs{\mathcal{A}} = \frac{\delta}{2}. \nonumber
\end{align}
This completes the proof.

\end{proof}

\subsection{Proof of Theorem~\ref{Thm:LearningProcess}}

We are now ready to prove Theorem~\ref{Thm:LearningProcess}.

\begin{proof}
Define $k^*$ as the number of required episodes for RL-QN to learn $\tilde{\pi}^*$. Based on Lemma \ref{Lem:Mix} and Lemma \ref{Lem:SampleNum}, we show that as the learning process proceeds, each $(\bm{Q}, a) \in \tilde{\mathcal{S}} \times \mathcal{A}$ will be sampled sufficiently to learn the optimal policy for $\tilde{\mathcal{M}}$. We define the event 
$$
B_k \triangleq \left\{ k \geqslant K_0, \tilde{\pi}_k = \pi_{\mathrm{rand}}, N_k (\bm{Q}, a) > \frac{p^{\mathrm{rand}} \left( \bm{Q} \right) \cdot L_k}{2 \abs{\mathcal{A}}}, \forall \bm{Q} \in \tilde{\mathcal{S}}, a \in \mathcal{A} \right\} .
$$

When $B_k$ is true, at least $p^{\mathrm{rand}} \left( \bm{Q} \right) L \sqrt{K_0} / (2 \abs{\mathcal{A}})$ samples are obtained for each $(\bm{Q}, a)$. Therefore, a sufficient condition to obtain the required number of samples for each $(\bm{Q}, a)$ as Lemma \ref{Lem:SampleNum} is that $B_k$ occurs for 
$$
J^*\triangleq \ceil*{\frac{4 \abs{\mathcal{A}} \log \frac{2^{R+1} (U+1)^D \abs{\mathcal{A}}}{\delta} }{L \sqrt{K_0} (\Delta p)^2 \cdot \min_{\bm{Q} \in \tilde{\mathcal{S}}}{p^{\mathrm{rand}} \left( \bm{Q} \right)}}} 
$$
times.

We denote by $m^*$ the number of episodes needed for $B_k$ to occur for $J^*$ times. Define $E_{k_1, k_2, \cdots, k_n}$ as the event that for episodes $k = k_1, k_2, \cdots, k_n$, the event $B_k$ does NOT occur. For $n \geqslant 1$, we have
\begin{align}
 & \Pr \left\{ m^* \geqslant K_0+J^*+n \right\} \nonumber \\
= & \Pr \left\{ \text{from episode } K_0 \text{ to } K_0+J^*+n-1 \text{, } B_k \text{ does not occur for at least } n \text{ times} \right\} \nonumber \\
= & \Pr \left\{ \bigcup_{ K_0 \leqslant k_1 < k_2 < \cdots < k_n \leqslant K_0+J^*+n-1} E_{k_1, k_2, \cdots, k_n} \right\} \nonumber \\
\leqslant & \sum_{K_0 \leqslant k_1 < k_2 < \cdots < k_n \leqslant K_0+J^*+n-1} \Pr \left\{ E_{k_1, k_2, \cdots, k_n} \right\}. \label{Eqn:Prm*1}
\end{align}

By applying Lemma \ref{Lem:Mix}, we have that for each $k \geqslant K_0$,
$$
\Pr \left\{ B_k \right\} \geqslant \frac{1}{2} \cdot \Pr \left\{ \pi_{\mathrm{rand}} \text{ is selected at episode } k \right\} = \frac{\ell}{2\sqrt{k}}.
$$
Therefore, for any $K_0 \leqslant k_1 < k_2 < \cdots < k_n \leqslant K_0+J^*+n-1$, we have
\begin{equation} \label{Eqn:Prm*2}
\Pr \left\{ E_{k_1, k_2, \cdots, k_n} \right\} \leqslant \prod_{i=1}^n \left( 1 - \frac{\ell}{2\sqrt{k_i}} \right) \leqslant \left( 1 - \frac{\ell}{2\sqrt{K_0+J^*+n-1}} \right)^n.
\end{equation}
Inserting \eqref{Eqn:Prm*2} into \eqref{Eqn:Prm*1} yields
\begin{align}
\Pr \left\{ m^* \geqslant K_0+J^*+n \right\} \leqslant & \binom{J^*+n}{n} \cdot \left( 1 - \frac{\ell}{2\sqrt{K_0+J^*+n-1}} \right)^n \nonumber \\
= & \frac{(J^*+n) \cdot (J^*+n-1) \cdots (n+1)}{J^*!} \cdot \left( 1 - \frac{\ell}{2\sqrt{K_0+J^*+n-1}} \right)^n \nonumber \\
\leqslant & \frac{(J^*+n)^{J^*}}{J^*!} \cdot \left( 1 - \frac{\ell}{2\sqrt{K_0+J^*+n-1}} \right)^n \nonumber \\
\leqslant & \frac{(J^*+n)^{J^*}}{J^*!} \cdot \exp \left( - \frac{n\ell}{2\sqrt{K_0+J^*+n-1}} \right) \label{Eqn:Prm*3} \\
< & \frac{(J^*+n)^{J^*}}{J^*!} \cdot \frac{(2J^*+4)! \cdot 4^{J^*+2} \cdot (K_0+J^*+n-1)^{J^*+2}}{(n\ell)^{2J^*+4}} \label{Eqn:Prm*4} \\
\leqslant & \frac{4^{J^*+2} \cdot (2J^*+4)! \cdot (K_0+J^*)^{2J^*+2}}{J^*! \cdot \ell^{2J^*+4}} \cdot \frac{1}{n^2} , \nonumber
\end{align}
where \eqref{Eqn:Prm*3} follows from the fact that for $x\in (0,1)$, $1-x\leq e^{-x}$ and \eqref{Eqn:Prm*4} is obtained by the fact that $\exp(u) = \sum_{k=0}^{\infty} (u)^k/k! > u^{2J^*+4}/(2J^*+4)!$ holds for $u>0$. 

We then have
\begin{align}
\mathbb{E} \left[ k^* \right] \leqslant & \mathbb{E} \left[ m^* \right] =  \sum_{i=1}^{\infty} \Pr \left\{ m^* \geqslant i \right\} = K_0 + J^* + \sum_{n=1}^{\infty} \Pr \left\{ m^* \geqslant K_0+J^*+n \right\} \nonumber \\
\leq  & K_0+J^* + \frac{4^{J^*+2} \cdot (2J^*+4)! \cdot (K_0+J^*)^{2J^*+2}}{J^*! \cdot \ell^{2J^*+4}} \sum_{n=1}^{\infty} \frac{1}{n^2} \nonumber\\
\leq & K_0+J^* + \frac{\pi^2 \cdot 4^{J^*+2} \cdot (2J^*+4)! \cdot (K_0+J^*)^{2J^*+2}}{6 \cdot J^*! \cdot \ell^{2J^*+4}} \triangleq K(J^*) . \nonumber
\end{align}
By Markov's inequality, we have
\begin{equation} \label{Eqn:ProbEpiNum}
k^* \leqslant \frac{2 \mathbb{E} \left[ k^* \right]}{\delta} \leqslant \frac{2 K(J^*)}{\delta},
\end{equation}
with probability at least $1 - \delta/2$.

By taking a union bound over the events of Lemma \ref{Lem:SampleNum} and \eqref{Eqn:ProbEpiNum}, we have that with probability at least $1-\delta$,  
$$
k^* \leqslant \frac{2}{\delta} \left( K_0+J^* + \frac{\pi^2 \cdot 4^{J^*+2} \cdot (2J^*+4)! \cdot (K_0+J^*)^{2J^*+2}}{6 \cdot J^*! \cdot \ell^{2J^*+4}} \right) ,
$$
which completes the proof.

\end{proof}


\section{Proof of Lemma \ref{Lem:EpiExBacklog}} \label{App:PfEpiExBacklog}

This section is devoted to the proof of Lemma~\ref{Lem:EpiExBacklog}. 

\begin{proof}
We only discuss an episode $k$ with $\pi_k^{\mathrm{in}} = \tilde{\pi}^*$ here.

We define $\sum_{t = t_{k-1}+1}^{t_k} \left( \sum_i Q_i(t) - \tilde{\rho}^* \right)$ as the regret of episode $k$. Let $t_k'$ and $t_k''$ denote the first and last time slot such that $\bm{Q}(t) \in \mathcal{S}^{\mathrm{in}}_{\mathrm{bd}}$ for $t\in\{t_{k-1}, \cdots, t_k\}$. For the simplicity, we first include the regret incurred at time step $t_{k-1}$ into the episodic regret analysis and subtract it in the end.

We decompose average episodic regret as follows.
\begin{align}
& \lim_{k \to \infty} \mathbb{E}_{\tilde{\pi}^* + \pi_0} \left[ \frac{\sum_{t = t_{k-1}+1}^{t_k} \left( \sum_i Q_i(t) - \tilde{\rho}^* \right)}{L_k'} \right] \nonumber \\
= & \lim_{k \to \infty} \mathbb{E}_{\tilde{\pi}^* + \pi_0} \left[ \frac{\sum_{t = t_k'}^{t_k''} \left( \sum_i Q_i(t) - \tilde{\rho}^* \right)}{L_k'} \right] + \label{Eqn:DecB} \\
& \lim_{k \to \infty} \mathbb{E}_{\tilde{\pi}^* + \pi_0} \left[ \frac{\sum_{t = t_{k-1}}^{t_k'-1} \left( \sum_i Q_i(t) - \tilde{\rho}^* \right) + \sum_{t = t_k''+1}^{t_k} \left( \sum_i Q_i(t) - \tilde{\rho}^* \right) - \left( \sum_i Q_i(t_{k-1}) - \tilde{\rho}^* \right)}{L_k'} \right] \label{Eqn:DecAC} .
\end{align}

\noindent{\bf Upper bound of \eqref{Eqn:DecB}:}

We partition the time slots from $t_k'$ to $t_k''$ into $\mathcal{T}^{\mathrm{in}}_k$, $\mathcal{T}^{\mathrm{bd}}_k$ and $\mathcal{T}^{\mathrm{out}}_k$, which denote the set of time slots that $\bm{Q}(t)$ is in $\mathcal{S}^{\mathrm{in}}_{\mathrm{in}}$, $\mathcal{S}^{\mathrm{in}}_{\mathrm{bd}}$ and $\mathcal{S}^{\mathrm{out}}$, respectively.

To analyze the regret associated with $\mathcal{T}^{\mathrm{in}}_k,$
we define an "in process" unit as the process that $\bm{Q}(t)$ leaves $\mathcal{S}^{\mathrm{in}}_{\mathrm{bd}}$, enters $\mathcal{S}^{\mathrm{in}}_{\mathrm{in}}$, stays in $\mathcal{S}^{\mathrm{in}}_{\mathrm{in}}$ for some time and finally returns back to $\mathcal{S}^{\mathrm{in}}_{\mathrm{bd}}$. Then the process during $\mathcal{T}^{\mathrm{in}}_k$ can be decomposed into multiple "in process" units. An "in process" unit is said to start from $\bm{Q}^{\mathrm{in}} \in \mathcal{S}^{\mathrm{in}}_{\mathrm{bd}}$ if $\bm{Q}^{\mathrm{in}}$ is its last state before entering  $\mathcal{S}^{\mathrm{in}}_{\mathrm{in}}$. We use $N^{\mathrm{in}}_k \big( \bm{Q}^{\mathrm{in}} \big)$ to denote the number of times that an "in process" unit starts from $\bm{Q}^{\mathrm{in}}$. The accumulated regret during the $i^{th}$ "in process" unit  starting from $\bm{Q}^{\mathrm{in}}$ is denoted by $R^{\mathrm{in}}_{k, i} \big( \bm{Q}^{\mathrm{in}} \big)$.

Similarly, we define "out process" units to decompose $\mathcal{T}^{\mathrm{out}}_k$.  An "out process" unit is said to start from $\bm{Q}^{\mathrm{out}}\in \mathcal{S}^{\mathrm{in}}_{\mathrm{bd}}$ if $\bm{Q}^{\mathrm{out}}$ is its last state before entering $\mathcal{S}^{\mathrm{out}}$. We define $N^{\mathrm{out}}_k \left( \bm{Q}^{\mathrm{out}} \right)$ and $R^{\mathrm{out}}_{k, i} \left( \bm{Q}^{\mathrm{out}} \right)$ in similar manners.

Now we can further decompose \eqref{Eqn:DecB} as follows.
\begin{align}
\lim_{k \to \infty} \mathbb{E}_{\tilde{\pi}^* + \pi_0} \left[ \frac{\sum_{t = t_k'}^{t_k''} \left( \sum_i Q_i(t) - \tilde{\rho}^* \right)}{L_k'} \right] = & \lim_{k \to \infty} \mathbb{E}_{\tilde{\pi}^* + \pi_0} \left[ \frac{\sum_{\bm{Q}^{\mathrm{in}} \in \mathcal{S}^{\mathrm{in}}_{\mathrm{bd}}} \sum_{i = 1}^{N^{\mathrm{in}}_k \left( \bm{Q}^{\mathrm{in}} \right)} R^{\mathrm{in}}_{k, i} \left( \bm{Q}^{\mathrm{in}} \right)}{L_k'} \right] + \label{Eqn:DecIn} \\
& \lim_{k \to \infty} \mathbb{E}_{\tilde{\pi}^* + \pi_0} \left[ \frac{\sum_{\bm{Q}^{\mathrm{out}} \in \mathcal{S}^{\mathrm{in}}_{\mathrm{bd}}} \sum_{i = 1}^{N^{\mathrm{out}}_k \left( \bm{Q}^{\mathrm{out}} \right)} R^{\mathrm{out}}_{k, i} \left( \bm{Q}^{\mathrm{out}} \right)}{L_k'} \right] + \label{Eqn:DecOut} \\
& \lim_{k \to \infty} \mathbb{E}_{\tilde{\pi}^* + \pi_0} \left[ \frac{\sum_{t \in \mathcal{T}^{\mathrm{bd}}_k} \left( \sum_i Q_i(t) - \tilde{\rho}^* \right)}{L_k'} \right] . \label{Eqn:DecBd}
\end{align}

For \eqref{Eqn:DecIn}, we use $Y^{\mathrm{in}}_{k, i} \big( \bm{Q}^{\mathrm{in}} \big)$ to denote the length of the time interval between the starting time of the $i^{th}$ and $(i+1)^{th}$ "in process" units that start from $\bm{Q}^{\mathrm{in}}$. By the Markovian property of the system, for a given $\bm{Q}^{\mathrm{in}}$, $Y^{\mathrm{in}}_{k, i} \big( \bm{Q}^{\mathrm{in}} \big)$'s are i.i.d. and $R^{\mathrm{in}}_{k, i} \big( \bm{Q}^{\mathrm{in}} \big)$'s are also i.i.d. Then according to the renewal reward theorem, for every $\bm{Q}^{\mathrm{in}} \in \mathcal{S}^{\mathrm{in}}_{\mathrm{bd}}$, we have 

\begin{equation} \label{Eqn:Ren1} 
\lim_{k \to \infty} \mathbb{E}_{\tilde{\pi}^* + \pi_0} \left[ \frac{ \sum_{i = 1}^{N^{\mathrm{in}}_k \left( \bm{Q}^{\mathrm{in}} \right)} R^{\mathrm{in}}_{k, i} \left( \bm{Q}^{\mathrm{in}} \right)}{L_k'} \right] = \frac{\mathbb{E}_{\tilde{\pi}^* + \pi_0} \left[ R^{\mathrm{in}}_{k, 1} \left( \bm{Q}^{\mathrm{in}} \right) \right]}{\mathbb{E}_{\tilde{\pi}^* + \pi_0} \left[ Y^{\mathrm{in}}_{k, 1} \left( \bm{Q}^{\mathrm{in}} \right) \right]} .
\end{equation}
{Note that we inherently use the fact that as $k \to \infty$, $L_k' \to \infty$ since $L_k' \geqslant L \sqrt{k}$.}

However, it is not straightforward to compute $\mathbb{E}_{\tilde{\pi}^* + \pi_0} \Big[ Y^{\mathrm{in}}_{k, 1} \big( \bm{Q}^{\mathrm{in}} \big) \Big]$  directly. Note that every time when an "in process" unit starting from $\bm{Q}^{\mathrm{in}}$ occurs, $\bm{Q}^{\mathrm{in}}$ must be visited. Therefore, we have the bound that for every $\bm{Q}^{\mathrm{in}} \in \mathcal{S}^{\mathrm{in}}_{\mathrm{bd}}$,
\begin{equation} \label{Eqn:Ren2}
\frac{1}{\mathbb{E}_{\tilde{\pi}^* + \pi_0} \left[ Y^{\mathrm{in}}_{k, 1} \left( \bm{Q}^{\mathrm{in}} \right) \right]} \leqslant \frac{1}{\mathbb{E}_{\tilde{\pi}^* + \pi_0} \left[ \text{Interval between visits to } \bm{Q}^{\mathrm{in}} \right]} = p^{\tilde{\pi}^* + \pi_0} \left( \bm{Q}^{\mathrm{in}} \right) .
\end{equation}

By inserting \eqref{Eqn:Ren1} and \eqref{Eqn:Ren2} into \eqref{Eqn:DecIn}, we have
\begin{align}
\lim_{k \to \infty} \mathbb{E}_{\tilde{\pi}^* + \pi_0} \left[ \frac{\sum_{\bm{Q}^{\mathrm{in}} \in \mathcal{S}^{\mathrm{in}}_{\mathrm{bd}}} \sum_{i = 1}^{N^{\mathrm{in}}_k \left( \bm{Q}^{\mathrm{in}} \right)} R^{\mathrm{in}}_{k, i} \left( \bm{Q}^{\mathrm{in}} \right)}{L_k'} \right] \leqslant & \sum_{\bm{Q}^{\mathrm{in}} \in \mathcal{S}^{\mathrm{in}}_{\mathrm{bd}}} p^{\tilde{\pi}^* + \pi_0} \left( \bm{Q}^{\mathrm{in}} \right) \cdot \mathbb{E}_{\tilde{\pi}^* + \pi_0} \left[ R^{\mathrm{in}}_{k, 1} \left( \bm{Q}^{\mathrm{in}} \right) \right] \nonumber \\
\leqslant & p^{\tilde{\pi}^* + \pi_0} \left( \mathcal{S}^{\mathrm{in}}_{\mathrm{bd}} \right) \cdot \max_{\bm{Q} \in \mathcal{S}^{\mathrm{in}}_{\mathrm{bd}}} \mathbb{E}_{\tilde{\pi}^* + \pi_0} \left[ R^{\mathrm{in}}_{k, 1} \left( \bm{Q} \right) \right] . \label{Eqn:DecIn2}
\end{align}

We provide an upper bound for $\max_{\bm{Q} \in \mathcal{S}^{\mathrm{in}}_{\mathrm{bd}}} \mathbb{E}_{\tilde{\pi}^* + \pi_0} \left[ R^{\mathrm{in}}_{k, 1} \left( \bm{Q} \right) \right]$, as stated in the following lemma (see Appendix \ref{App:BoundRIn} for the proof). 

\begin{lemma} \label{Lem:BoundRIn} Under Assumptions~\ref{Asp:KnownPolicy}-\ref{Asp:PolyTrans}, we have 
$$
\max_{\bm{Q} \in \mathcal{S}^{\mathrm{in}}_{\mathrm{bd}}} \mathbb{E}_{\tilde{\pi}^* + \pi_0} \left[ R^{\mathrm{in}}_{k, 1} \left( \bm{Q} \right) \right] \leqslant cDU^{1+\gamma}.
$$
\end{lemma}

For \eqref{Eqn:DecOut}, by following a  similar argument, we have that
\begin{equation} \label{Eqn:DecOut2}
\lim_{k \to \infty} \mathbb{E}_{\tilde{\pi}^* + \pi_0} \left[ \frac{\sum_{\bm{Q}^{\mathrm{out}} \in \mathcal{S}^{\mathrm{in}}_{\mathrm{bd}}} \sum_{i = 1}^{N^{\mathrm{out}}_k \left( \bm{Q}^{\mathrm{out}} \right)} R^{\mathrm{out}}_{k, i} \left( \bm{Q}^{\mathrm{out}} \right)}{L_k'} \right] \leqslant p^{\tilde{\pi}^* + \pi_0} \left( \mathcal{S}^{\mathrm{in}}_{\mathrm{bd}} \right) \cdot \max_{\bm{Q} \in \mathcal{S}^{\mathrm{in}}_{\mathrm{bd}}} \mathbb{E}_{\tilde{\pi}^* + \pi_0} \left[ R^{\mathrm{out}}_{k, 1} \left( \bm{Q} \right) \right].
\end{equation}

The following lemma gives an upper bound for $\max_{\bm{Q} \in \mathcal{S}^{\mathrm{in}}_{\mathrm{bd}}} \mathbb{E}_{\tilde{\pi}^* + \pi_0} \left[ R^{\mathrm{out}}_{k, 1} \left( \bm{Q} \right) \right]$ (see Appendix \ref{App:BoundROut} for the proof).

\begin{lemma} \label{Lem:BoundROut} Under Assumptions~\ref{Asp:KnownPolicy}-\ref{Asp:PolyTrans}, we have
$$
\max_{\bm{Q} \in \mathcal{S}^{\mathrm{in}}_{\mathrm{bd}}} \mathbb{E}_{\tilde{\pi}^* + \pi_0} \left[ R^{\mathrm{out}}_{k, 1} \left( \bm{Q} \right) \right] \leqslant \frac{2a^2DU^{2\alpha+1}}{\epsilon_0^2} .
$$
\end{lemma}

For \eqref{Eqn:DecBd}, since under $\tilde{\pi}^* + \pi_0$, the Markov chain is positive recurrent, we have that
\begin{equation} \label{Eqn:DecBd2} 
\lim_{k \to \infty} \mathbb{E}_{\tilde{\pi}^* + \pi_0} \left[ \frac{\sum_{t \in \mathcal{T}^{\mathrm{bd}}_k} \left( \sum_i Q_i(t) - \tilde{\rho}^* \right)}{L_k'} \right] \leqslant (DU-\tilde{\rho}^*) \cdot \lim_{k \to \infty} \mathbb{E}_{\tilde{\pi}^* + \pi_0} \left[ \frac{\abs{\mathcal{T}^{\mathrm{bd}}_k}}{L_k'}  \right] \leqslant p^{\tilde{\pi}^* + \pi_0} \left( \mathcal{S}^{\mathrm{in}}_{\mathrm{bd}} \right) \cdot DU .
\end{equation}

By combining \eqref{Eqn:DecIn2}, \eqref{Eqn:DecOut2}, \eqref{Eqn:DecBd2}, Lemma \ref{Lem:BoundRIn} and Lemma \ref{Lem:BoundROut}, we upper bound \eqref{Eqn:DecB} as follows:
\begin{equation} \label{Eqn:DecB2}
\lim_{k \to \infty} \mathbb{E}_{\tilde{\pi}^* + \pi_0} \left[ \frac{\sum_{t = t_k'}^{t_k''} \left( \sum_i Q_i(t) - \tilde{\rho}^* \right)}{L_k'} \right] \leqslant p^{\tilde{\pi}^* + \pi_0} \left( \mathcal{S}^{\mathrm{in}}_{\mathrm{bd}} \right) \cdot \mathcal{O} \left( U^{1 + \max \{ 2\alpha, \gamma \}} \right) .
\end{equation}

\noindent{\bf Upper bound of \eqref{Eqn:DecAC}:}

From the episode termination criteria of RL-QN, we have $\bm{Q}(t_{k-1}) \in \mathcal{S}^{\mathrm{in}}$. If $\bm{Q}(t_{k-1}) \in \mathcal{S}^{\mathrm{in}}_{\mathrm{bd}}$, we have that $t_k' = t_{k-1}$, and therefore have that
\begin{equation*} 
\mathbb{E}_{\tilde{\pi}^* + \pi_0} \left[ \sum_{t = t_{k-1}}^{t_k'-1} \left( \sum_i Q_i(t) - \tilde{\rho}^* \right) \mid \bm{Q}(t_{k-1}) \in \mathcal{S}^{\mathrm{in}}_{\mathrm{bd}} \right] = 0.
\end{equation*}
If $\bm{Q}(t_{k-1}) \in \mathcal{S}^{\mathrm{in}}_{\mathrm{in}}$, then from $t = t_{k-1}$ to $t = t_k'-1$, $\bm{Q}(t) \in \mathcal{S}^{\mathrm{in}}_{\mathrm{in}}$. By applying techniques in the proof of Lemma \ref{Lem:BoundRIn}, we have 
\begin{equation*} 
\mathbb{E}_{\tilde{\pi}^* + \pi_0} \left[ \sum_{t = t_{k-1}}^{t_k'-1} \left( \sum_i Q_i(t) - \tilde{\rho}^* \right) \mid \bm{Q}(t_{k-1}) \in \mathcal{S}^{\mathrm{in}}_{\mathrm{in}} \right] \leqslant cDU^{1+\gamma} .
\end{equation*}

Therefore, we have 
\begin{equation} \label{Eqn:DecAC4}
\mathbb{E}_{\tilde{\pi}^* + \pi_0} \left[ \sum_{t = t_{k-1}}^{t_k'-1} \left( \sum_i Q_i(t) - \tilde{\rho}^* \right) \right] \leqslant cDU^{1+\gamma} .
\end{equation}

Since we also have that $\bm{Q}(t_k) \in \mathcal{S}^{\mathrm{in}}$, by following a similar argument, we have 
\begin{equation} \label{Eqn:DecAC5}
\mathbb{E}_{\tilde{\pi}^* + \pi_0} \left[ \sum_{t = t_k''+1}^{t_k} \left( \sum_i Q_i(t) - \tilde{\rho}^* \right) \right] \leqslant cDU^{1+\gamma} .
\end{equation}

By combining \eqref{Eqn:DecAC4} and \eqref{Eqn:DecAC5}, together with the fact that $L_k' \geqslant L_k$, we have an upper bound for \eqref{Eqn:DecAC} as follows.
\begin{align}
& \lim_{k \to \infty} \mathbb{E}_{\tilde{\pi}^* + \pi_0} \left[ \frac{\sum_{t = t_{k-1}}^{t_k'-1} \left( \sum_i Q_i(t) - \tilde{\rho}^* \right) + \sum_{t = t_k''+1}^{t_k} \left( \sum_i Q_i(t) - \tilde{\rho}^* \right) - \left( \sum_i Q_i(t_{k-1}) - \tilde{\rho}^* \right)}{L_k'} \right] \nonumber \\
\leqslant & \lim_{k \to \infty} \frac{1}{L_k} \mathbb{E}_{\tilde{\pi}^* + \pi_0} \left[\sum_{t = t_{k-1}}^{t_k'-1} \left( \sum_i Q_i(t) - \tilde{\rho}^* \right) + \sum_{t = t_k''+1}^{t_k} \left( \sum_i Q_i(t) - \tilde{\rho}^* \right) - \left( \sum_i Q_i(t_{k-1}) - \tilde{\rho}^* \right) \right] \nonumber \\
\leqslant & \lim_{k \to \infty} \frac{2cDU^{1+\gamma} + \tilde{\rho}^*}{L \cdot \sqrt{k}} = 0 . \label{Eqn:DecAC6}
\end{align}

By summing up \eqref{Eqn:DecB2} and \eqref{Eqn:DecAC6}, we complete the proof.

\end{proof}


\subsection{Proof of Lemma \ref{Lem:BoundRIn}} \label{App:BoundRIn}

\begin{proof}

From Proposition 5.5.1 in \cite{bertsekas2017dynamic}, when applying $\tilde{\pi}^*$ to $\tilde{\mathcal{M}}$, there exists a function $\tilde{h}^*:\tilde{\mathcal{S}}\rightarrow \mathbb{R}$ such that  the following Bellman equation holds:
\begin{equation} \label{Eqn:Bell_1}
\tilde{\rho}^* + \tilde{h}^* \left( \bm{Q} \right) = \sum_i Q_i + \sum_{\bm{Q}' \in \mathcal{S}^{\mathrm{in}}} \tilde{p} \left( \bm{Q}' \mid \bm{Q}, \tilde{\pi}^* \left( \bm{Q} \right) \right) \cdot \tilde{h}^* \left( \bm{Q}' \right),\qquad \forall \bm{Q} \in \mathcal{S}^{\mathrm{in}}.
\end{equation}

By the construction of $\tilde{\mathcal{M}}$ in Section \ref{Sec:Approach}, for each $\bm{Q} \in \mathcal{S}^{\mathrm{in}}_{\mathrm{in}}$, $\bm{Q}' \in \mathcal{S}^{\mathrm{in}}$ and $a \in \mathcal{A}$, we have
$$
\tilde{p} \left( \bm{Q}' \mid \bm{Q}, \tilde{\pi}^* \left( \bm{Q} \right) \right) = p \left( \bm{Q}' \mid \bm{Q}, \tilde{\pi}^* \left( \bm{Q} \right) \right).
$$ 
Therefore, for each $\bm{Q} \in \mathcal{S}^{\mathrm{in}}_{\mathrm{in}}$, \eqref{Eqn:Bell_1} can be rewritten as
\begin{equation} \label{Eqn:Bell_2}
\tilde{\rho}^* + \tilde{h}^* \left( \bm{Q} \right) = \sum_i Q_i + \sum_{\bm{Q}' \in \mathcal{S}^{\mathrm{in}}} p \left( \bm{Q}' \mid \bm{Q}, \tilde{\pi}^* \left( \bm{Q} \right) \right) \cdot \tilde{h}^* \left( \bm{Q}' \right).
\end{equation}

Fix an $\hat{\bm{Q}} \in \mathcal{S}^{\mathrm{in}}_{\mathrm{bd}}$, we now analyze $R^{\mathrm{in}}_{k, 1} \left( \hat{\bm{Q}} \right)$. We denote the start and end time slot of this "in process" unit as $t_s$ and $t_e$ respectively. Note that $\bm{Q}(t)\in \mathcal{S}^{\mathrm{in}}_{\mathrm{in}}$ for $t\in [t_s,t_e]$. Using \eqref{Eqn:Bell_2}, we have 

\begin{align}
\mathbb{E}_{\tilde{\pi}^* + \pi_0} \left[ R^{\mathrm{in}}_{k, 1} \left( \hat{\bm{Q}} \right) \right] = & \mathbb{E} \left[ \sum_{t = t_s}^{t_e} \left( \sum_i Q_i - \tilde{\rho}^* \right) 
\mid \bm{Q}(t_s-1)=\hat{\bm{Q}} \right] \nonumber \\
= & \mathbb{E}_{\tilde{\pi}^* + \pi_0} \left[ \sum_{t = t_s}^{t_e} \left( \tilde{h}^* \left( \bm{Q}(t) \right) - \mathbb{E} \left[ \tilde{h}^* \left( \bm{Q}(t+1) \right) \mid \bm{Q}(t) \right] \right) \mid \bm{Q}(t_s-1)=\hat{\bm{Q}}  \right] \nonumber \\
= & \mathbb{E}_{\tilde{\pi}^* + \pi_0} \left[ \tilde{h}^* \left(\bm{Q}(t_s)\right) - \mathbb{E} \left[ \tilde{h}^* \left( \bm{Q} (t_e + 1) \right) \mid \bm{Q}(t_e) \right] \mid \bm{Q}(t_s-1)=\hat{\bm{Q}}  \right] \nonumber \\
& +  \mathbb{E}_{\tilde{\pi}^* + \pi_0} \left[ \sum_{t = t_s}^{t_e - 1} \left( \tilde{h}^* (\bm{Q}(t+1)) - \mathbb{E} \left[ \tilde{h}^* (\bm{Q}(t+1)) \mid \bm{Q}(t) \right] \right) \mid \bm{Q}(t_s-1)=\hat{\bm{Q}}  \right] \nonumber \\
= & \mathbb{E}_{\tilde{\pi}^* + \pi_0} \left[ \tilde{h}^* \left(\bm{Q}(t_s)\right) - \tilde{h}^* \left( \bm{Q} (t_e + 1) \right) \mid \bm{Q}(t_s-1)=\hat{\bm{Q}}  \right] . \label{Eqn:DiffH}
\end{align}

On the other hand, for each $(\bm{Q}, \bm{Q}') \in \mathcal{S}\times \mathcal{S}$, by following the analysis for the proof of Proposition 5.5.1 in \cite{bertsekas2017dynamic}, we can upper bound $\tilde{h}^* \left( \bm{Q}' \right) - \tilde{h}^* \left( \bm{Q} \right)$ as follows,
\begin{equation} \label{Eqn:BoundH}
\tilde{h}^* \left( \bm{Q}' \right) - \tilde{h}^* \left( \bm{Q} \right) = \min_{\tilde{\pi}} \tilde{\mathbb{E}}_{\tilde{\pi}} \left[ \sum_{t=1}^{T_{\bm{Q} \to \bm{Q}'}} \left[ \sum_{i}Q_i(t) - \tilde{\rho}^* \right] \right] 
\stackrel{(a)}{\leqslant}
\min_{\tilde{\pi}} \tilde{\mathbb{E}}_{\tilde{\pi}} \left[ T_{\bm{Q} \to \bm{Q}'} \right] \cdot DU \stackrel{(b)}{\leqslant} cDU^{1+\gamma}, 
\end{equation}
where (a) comes from the fact that $\sum_{i}Q_i(t) \leqslant DU$ for  $\bm{Q}\in \tilde{\mathcal{S}}$, and the second inequality (b) holds under Assumption~\ref{Asp:PolyTrans}. 
By inserting \eqref{Eqn:BoundH} into \eqref{Eqn:DiffH}, we complete the proof.

\end{proof}


\subsection{Proof of Lemma \ref{Lem:BoundROut}} \label{App:BoundROut}

\begin{proof}
Fix an $\hat{\bm{Q}} \in \mathcal{S}^{\mathrm{in}}_{\mathrm{bd}}$, we now turn to analyze $R^{\mathrm{out}}_{k, 1} \left( \hat{\bm{Q}} \right)$. From the definition of $\mathcal{S}^{\mathrm{in}}$, we have $\hat{Q}_{\max} \leqslant U-W$. We denote the length of this "out process" unit as $\tau$. We then have 

\begin{equation} \label{Eqn:BoundROut}
\mathbb{E}_{\tilde{\pi}^* + \pi_0} \left[ R^{\mathrm{out}}_{k, 1} \left( \hat{\bm{Q}} \right) \right] \leqslant \mathbb{E}_{\tilde{\pi}^* + \pi_0} \left[ D \cdot \left( U-W + W \tau \right) \cdot \tau \right] = D(U-W) \mathbb{E}_{\tilde{\pi}^* + \pi_0} \left[ \tau \right] + DW \mathbb{E}_{\tilde{\pi}^* + \pi_0} \left[ \tau^2 \right] ,
\end{equation}
where the inequality comes from the fact that the backlog of each queue can be no larger than $\left( U - W + W \tau \right)$ during this "out process" unit. We upper bound $\mathbb{E}_{\tilde{\pi}^* + \pi_0}[\tau]$ and $\mathbb{E}_{\tilde{\pi}^* + \pi_0}[\tau^2]$ separately.

\medskip
\noindent{\bf Upper bound of $\mathbb{E}_{\tilde{\pi}^* + \pi_0} \left[ \tau \right]$:} We use the following lemma from \cite{bremaud1999lyapunov}.

\begin{lemma} [Theorem 1.1 in Chapter 5 of \cite{bremaud1999lyapunov}] \label{Lem:BoundTau}
For an irreducible Markov chain with countable state space, if there exists $\epsilon > 0$, a finite state set $\mathcal{F}$ and a nonnegative Lyapunov function $\Phi \left( \cdot \right)$, such that 
$\mathbb{E} \left[\Phi \left( \bm{Q}(t+1) \right) \mid \bm{Q}(t)=\bm{Q} \right] < \infty$ for each $\bm{Q} \notin \mathcal{F}$ and $\mathbb{E} \left[\Phi \left( \bm{Q}(t+1) \right) - \Phi \left( \bm{Q}(t) \right) \mid \bm{Q}(t)=\bm{Q} \right] \leqslant -\epsilon$ for each $\bm{Q} \notin \mathcal{F}$, we then have that for each $\bm{Q} \notin \mathcal{F}$, $\mathbb{E} \left[ T_{\bm{Q} \to \mathcal{F}} \right] \leqslant \Phi \left( \bm{Q} \right) / \epsilon$, where $T_{\bm{Q} \to \mathcal{F}}$ is the first hitting time of the set $\mathcal{F}$ when starting from state $\bm{Q}$.
\end{lemma}

We can interpret $\tau$ as the hitting time of the set $\mathcal{S}^{\mathrm{in}}$ when the Markov chain starts from state $\bm{Q} \in \mathcal{S}^{\mathrm{out}}$ . From the selection of $U$ in Algorithm $\ref{Alg:RL-QN}$, each state in $\mathcal{S}^{\mathrm{out}}$ has negative Lyapunov drift regarding $\Phi_0$. Also, from Assumption~\ref{Asp:KnownPolicy}, the value of Lyapunov function $\Phi_0$ for the beginning state of the "out process" unit can be upper bounded as $a (U-W+W)^{\alpha} = a U^{\alpha}$. By Lemma~\ref{Lem:BoundTau}, it follows that
\begin{equation} \label{Eqn:BoundTau}
\mathbb{E}_{\tilde{\pi}^* + \pi_0} \left[ \tau  \right] \leqslant \frac{a U^{\alpha}}{\epsilon_0} \triangleq T_0.
\end{equation}

\medskip
\noindent{\bf Upper bound of $\mathbb{E}_{\tilde{\pi}^* + \pi_0} \left[ \tau^2 \right]$:} We use the following lemma from \cite{hou2012homogeneous}.

\begin{lemma} [Theorem 6.3.4 in \cite{hou2012homogeneous}] \label{Lem:Moments}
In a Markov chain with a countable state space, for a nonempty state set $\mathcal{B}$, and a state $s$, if there exists a constant $C$ such that $\mathbb{E} \left[ T_{s \to \mathcal{B}} \right] \leqslant C \cdot F^*_{s, \mathcal{B}}$, 
then for $p \geqslant 1$, we have $\mathbb{E} \left[ T_{s \to \mathcal{B}}^p \right] \leqslant p! \cdot C^p \cdot F^*_{s, \mathcal{B}}$, where $F^*_{s, \mathcal{B}} \triangleq \Pr\{ s_1 \in \mathcal{B} \mid s_0 = s \} + \sum_{n=2}^{\infty} \Pr\{ s_1 \notin \mathcal{B}, \cdots, s_{n-1} \notin \mathcal{B}, s_n \in \mathcal{B} \mid s_0 = s \}$.
\end{lemma}

In our case, under $\tilde{\pi}^* + \pi_0$, the Markov chain is positive recurrent. Thus, for each $\bm{Q} \in \mathcal{S}^{\mathrm{out}}$ and each $\bm{Q}' \in \mathcal{S}^{\mathrm{in}}$, we have $F^*_{\bm{Q}, \bm{Q}'} = 1$. Note that $F^*_{\bm{Q}, \mathcal{S}^{\mathrm{in}}}$ is upper bounded by $1$ since it is a probability. Also, $F^*_{\bm{Q}, \mathcal{S}^{\mathrm{in}}} \geqslant F^*_{\bm{Q}, \bm{Q}'}$, 
we thus have $F^*_{\bm{Q}, \mathcal{S}^{\mathrm{in}}} = 1$. Therefore, $\mathbb{E} \left[ \tau  \right] \leqslant T_0 = T_0 \cdot F^*_{\bm{Q}, \mathcal{S}^{\mathrm{in}}}$. By applying Lemma \ref{Lem:Moments}, we have that

\begin{equation} \label{Eqn:BoundTau^2}
\mathbb{E}_{\tilde{\pi}^* + \pi_0} \left[ \tau^2 \right] \leqslant 2\cdot T_0^2 \cdot F^*_{\bm{Q}, \mathcal{S}^{\mathrm{in}}} = 2 T_0^2.
\end{equation}

Inserting \eqref{Eqn:BoundTau} and \eqref{Eqn:BoundTau^2} into \eqref{Eqn:BoundROut} gives
$$
\mathbb{E}_{\tilde{\pi}^* + \pi_0} \left[ R^{\mathrm{out}}_{k, 1} \left( \hat{\bm{Q}} \right) \right] \leqslant D(U-W)T_0 + 2DWT_0^2 \leqslant 2D(U-W+W)T_0^2 = \frac{2a^2DU^{2\alpha+1}}{\epsilon_0^2} ,
$$
which completes the proof.

\end{proof}


\section{Proof of Lemma \ref{Lem:pBd}} \label{App:pBd}

\begin{proof}
For the ease of exposition, we first define a ``linear'' Lyapunov function $\Phi'(\cdot)$ as
$$
\Phi' \left( \bm{Q} \right) = 
\begin{cases}
\left[ \tilde{\Phi}^* \left( \bm{Q} \right) \right]^{\frac{1}{\beta}}, & \text{ for } \bm{Q} \in \mathcal{S}^{\mathrm{in}} \\
U-W, & \text{ for } \bm{Q} \in \mathcal{S}^{\mathrm{out}}
\end{cases}
$$
which has the following properties (see Appendix \ref{App:LinearPhi} for the proof).

\begin{lemma} \label{Lem:LinearPhi} Under Assumption~\ref{Asp:DriftOptimalPolicy},
there exist constants $V, \epsilon' > 0$ such that for each $U > W$, the following properties hold:
\begin{enumerate}
\item For each $\bm{Q} \in \mathcal{S}$, $\max_{\bm{Q}' \in \mathcal{R} \left( \bm{Q} \right)} \abs{\Phi' \left( \bm{Q}' \right) - \Phi' \left( \bm{Q} \right)} \leqslant V$;
\item For each $\bm{Q} \in \mathcal{S}^{\mathrm{in}}$ such that $Q_{\max} \geqslant \tilde{B}^*$, we have
$$\mathbb{E}_{\tilde{\pi}^* + \pi_0} \left[\Phi' \left( \bm{Q}(t+1) \right) - \Phi' \left( \bm{Q}(t) \right) \mid \bm{Q}(t) = \bm{Q} \right] \leqslant -\epsilon'.$$

\end{enumerate}

\end{lemma}

With Lemma \ref{Lem:LinearPhi}, we obtain the following properties on the distribution of $\Phi'(\cdot)$ (see Appendix \ref{App:TailProb} for the proof).

\begin{lemma} \label{Lem:TailProb} Suppose Assumption~\ref{Asp:DriftOptimalPolicy} holds.
Define $m^* \triangleq \left\lfloor (U+V-W-b_1^{1/\beta}\tilde{B}^*)/(2V) \right\rfloor$, we then have that for $m = 1, 2, \cdots , m^*$,
\begin{align}
& p^{\tilde{\pi}^* + \pi_0} \left( \left\{ \bm{Q} \in \mathcal{S}^{\mathrm{in}} : \Phi' \left( \bm{Q} \right) > U-W-(2m-1)V \right\} \right) \nonumber \\
\leqslant & \frac{V}{V+\tilde{\epsilon}^*} \cdot p^{\tilde{\pi}^* + \pi_0} \left( \left\{ \bm{Q} \in \mathcal{S}^{\mathrm{in}} : \Phi' \left( \bm{Q} \right) > U-W-(2m+1)V \right\} \right) \nonumber .
\end{align} 
\end{lemma}

Lemma \ref{Lem:TailProb} implies that
\begin{align*}
&p^{\tilde{\pi}^* + \pi_0} \left( \left\{\bm{Q} \in \mathcal{S}^{\mathrm{in}} : \Phi' \left( \bm{Q} \right) > U-W-V \right\} \right) \\
\leqslant & \left( \frac{V}{V+\epsilon^*} \right)^{m^*} \cdot p^{\tilde{\pi}^* + \pi_0} \left( \left\{ \bm{Q} \in \mathcal{S}^{\mathrm{in}} : \Phi' \left( \bm{Q} \right) > U-W-(2m^*+1)V \right\} \right) \leqslant \left( \frac{V}{V+\epsilon^*} \right)^{m^*}.
\end{align*}

Note that our goal is to obtain the stationary probability of the state being in $\mathcal{S}^{\mathrm{in}}_{\mathrm{bd}}$.
Recall that $\mathcal{S}^{\mathrm{in}} \triangleq \big\{ \bm{Q} : \tilde{\Phi}^* \left( \bm{Q} \right) \leqslant (U-W)^\beta \big\}$
and
$\mathcal{S}^{\mathrm{in}}_{\mathrm{bd}} \triangleq \left\{ \bm{Q} \in \mathcal{S}^{\mathrm{in}}: \mathcal{R} \left( \bm{Q} \right) \cap \mathcal{S}^{\mathrm{out}} \neq \emptyset \right\}$.
Therefore,  
$$
\max_{\bm{Q}' \in \mathcal{R}(\bm{Q})} \tilde{\Phi}^* \left( \bm{Q}' \right) > (U-W)^{\beta},\qquad \forall \bm{Q} \in \mathcal{S}^{\mathrm{in}}_{\mathrm{bd}}.
$$
By the definition of $\Phi'$ and Lemma~\ref{Lem:LinearPhi}, we have $\Phi' \left( \bm{Q} \right) > U-W-V$.  Therefore,
$$
\mathcal{S}^{\mathrm{in}}_{\mathrm{bd}} \subseteq \left\{ \bm{Q} \in \mathcal{S}^{\mathrm{in}} : \Phi' \left( \bm{Q} \right) > U-W-V \right\}
$$ 
and the following inequality holds
$$
p^{\tilde{\pi}^* + \pi_0} \left( \mathcal{S}^{\mathrm{in}}_{\mathrm{bd}} \right) \leqslant p^{\tilde{\pi}^* + \pi_0} \left(  \left\{ \bm{Q} \in \mathcal{S}^{\mathrm{in}} : \Phi' \left( \bm{Q} \right) > U-W-V \right\} \right) \leqslant \left( \frac{V}{V+\epsilon^*} \right)^{m^*} .
$$

By taking natural logarithm on both sides and applying the fact that $\log(1+x) \geqslant x/(1+x)$ for $x > 0$, we have
$$
\log p^{\tilde{\pi}^* + \pi_0} \left( \mathcal{S}^{\mathrm{in}}_{\mathrm{bd}} \right) \leqslant -m^* \log \left( 1+\frac{\tilde{\epsilon}^*}{V} \right) \leqslant -m^* \cdot \frac{\tilde{\epsilon}^*}{V+\tilde{\epsilon}^*} = - \left\lfloor \frac{U+V-W-b_1^{1/\beta}\tilde{B}^*}{2V} \right\rfloor \cdot \frac{\tilde{\epsilon}^*}{V+\tilde{\epsilon}^*} = \mathcal{O} \left( -U \right) ,
$$
which completes the proof.

\end{proof}

\subsection{Proof of Lemma \ref{Lem:LinearPhi}} \label{App:LinearPhi}

\begin{proof}
Consider a fixed $\bm{Q} \in \mathcal{S}^{\mathrm{in}}.$
From the definition of $\mathcal{S}^{\mathrm{in}}$ and $W$,  we have  $\max_{\bm{Q}' \in \mathcal{R} \left( \bm{Q} \right)} {Q}'_{\max} \leqslant U.$ Hence  $\bm{Q}' \in \tilde{\mathcal{S}}$ for each $\bm{Q}' \in \mathcal{R} \left( \bm{Q} \right)$. By the definition of function $\Phi'(\cdot),$ we have
\begin{align}
& \mathbb{E}_{\tilde{\pi}^* + \pi_0} \left[\Phi' \left( \bm{Q}(t+1) \right) - \Phi' \left( \bm{Q}(t) \right) \mid \bm{Q}(t) = \bm{Q} \right] \nonumber \\
= & - \left[ \tilde{\Phi}^* \left( \bm{Q} \right) \right]^{\frac{1}{\beta}} + \sum_{\bm{Q}' \in \mathcal{R} \left( \bm{Q} \right)} \tilde{p} \left( \bm{Q}' \mid \bm{Q}, \tilde{\pi}^*(\bm{Q}) \right) \cdot \Phi' \left( \bm{Q}' \right) \label{Eqn:LP_1} . 
\end{align}

Note that
\begin{align}
& \sum_{\bm{Q}' \in \mathcal{R} \left( \bm{Q} \right)} \tilde{p} \left( \bm{Q}' \mid \bm{Q}, \tilde{\pi}^*(\bm{Q}) \right) \cdot \Phi' \left( \bm{Q}' \right) \nonumber \\
= & \sum_{\bm{Q}' \in \mathcal{R} \left( \bm{Q} \right) \cap \mathcal{S}^{\mathrm{in}}} \tilde{p} \left( \bm{Q}' \mid \bm{Q}, \tilde{\pi}^*(\bm{Q}) \right) \cdot \left[ \tilde{\Phi}^* \left( \bm{Q}' \right) \right]^{\frac{1}{\beta}} + \sum_{\bm{Q}' \in \mathcal{R} \left( \bm{Q} \right) \cap \mathcal{S}^{\mathrm{out}}} \tilde{p} \left( \bm{Q}' \mid \bm{Q}, \tilde{\pi}^*(\bm{Q}) \right) \cdot (U-W) \nonumber \\
\leqslant & \sum_{\bm{Q}' \in \mathcal{R} \left( \bm{Q} \right)} \tilde{p} \left( \bm{Q}' \mid \bm{Q}, \tilde{\pi}^*(\bm{Q}) \right) \cdot \left[ \tilde{\Phi}^* \left( \bm{Q}' \right) \right]^{\frac{1}{\beta}} , \label{Eqn:LP_2}
\end{align}
where the inequality comes from the fact that $\left[ \tilde{\Phi}^* \left( \bm{Q}' \right) \right]^{1 / \beta} \geqslant U-W$ for $\bm{Q}' \in \mathcal{S}^{\mathrm{out}}$.

Combining Eqs \eqref{Eqn:LP_1} and \eqref{Eqn:LP_2} yields
\begin{align}
& \mathbb{E}_{\tilde{\pi}^* + \pi_0} \left[\Phi' \left( \bm{Q}(t+1) \right) - \Phi' \left( \bm{Q}(t) \right) \mid \bm{Q}(t) = \bm{Q} \right] \nonumber \\
\leqslant & \tilde{\mathbb{E}}_{\tilde{\pi}^*} \left[\left( \tilde{\Phi}^* \left( \bm{Q}(t+1) \right) \right)^{\frac{1}{\beta}} \mid \bm{Q}(t) = \bm{Q} \right] - \left( \tilde{\Phi}^* \left( \bm{Q} \right) \right)^{\frac{1}{\beta}} \nonumber \\
\stackrel{(a)}{\leqslant} & \left( \tilde{\mathbb{E}}_{\tilde{\pi}^*} \left[ \tilde{\Phi}^* \left( \bm{Q}(t+1) \right) \mid \bm{Q}(t) = \bm{Q} \right] \right)^{\frac{1}{\beta}} - \left( \tilde{\Phi}^* \left( \bm{Q} \right) \right)^{\frac{1}{\beta}} \nonumber \\
\stackrel{(b)}{\leqslant} & \frac{\left( \tilde{\Phi}^* \left( \bm{Q} \right) \right)^{\frac{1}{\beta} - 1}}{\beta} \cdot \left( \tilde{\mathbb{E}}_{\tilde{\pi}^*} \left[ \tilde{\Phi}^* \left( \bm{Q}(t+1) \right) \mid \bm{Q}(t) = \bm{Q} \right] - \tilde{\Phi}^* \left( \bm{Q} \right) \right) \nonumber \\
\stackrel{(c)}{\leqslant} & - \frac{\left( \tilde{\Phi}^* \left( \bm{Q} \right) \right)^{\frac{1}{\beta} - 1}}{\beta} \cdot \tilde{\epsilon}^* \cdot Q_{\max}^{\beta-1} \stackrel{(d)}{\leqslant} - \frac{b_1^{\frac{1}{\beta} - 1}}{\beta} \cdot \tilde{\epsilon}^* \triangleq -\epsilon', \nonumber 
\end{align}
where (a) follows from Jensen's inequality, (b) holds because $f(x) = x^{1/\beta}$ is concave and thus $f(y) - f(x) \leqslant f'(x) (y-x)$, (c) and (d) come from Assumption~\ref{Asp:DriftOptimalPolicy}. 

Next we discuss the upper bound of $\max_{\bm{Q}' \in \mathcal{R} \left( \bm{Q} \right)} \abs{\Phi' \left( \bm{Q}' \right) - \Phi' \left( \bm{Q} \right)}$.

If $\bm{Q}, \bm{Q}' \in \mathcal{S}^{\mathrm{in}}$ with $Q_{\max} = 0$ or $Q'_{\max} = 0$, we have $\abs{\Phi' \left( \bm{Q}' \right) - \Phi' \left( \bm{Q} \right)} \leqslant b_1^{1/\beta} W$.

If $\bm{Q}, \bm{Q}'  \in \mathcal{S}^{\mathrm{in}}$ with $Q_{\max} > 0$ and $Q'_{\max} > 0,$ we have $\Phi' \left( \bm{Q}' \right)=\left( \tilde{\Phi}^* \left( \bm{Q}' \right) \right)^{\frac{1}{\beta}}$ and $\Phi' \left( \bm{Q} \right)=\left( \tilde{\Phi}^* \left( \bm{Q} \right) \right)^{\frac{1}{\beta}}.$ 
Note that for the concave function $f(x)=x^{1/\beta}$, it holds that
$$
\abs{f(y) - f(x)} \leqslant \max \{ f'(x), f'(y) \} \cdot \abs{y-x} .
$$
Without losing generality, we assume that $Q'_{\max} > Q_{\max}.$ Thus
\begin{align}
\abs{\Phi' \left( \bm{Q}' \right) - \Phi' \left( \bm{Q} \right)} \leqslant & \max \Big\{ \left( \tilde{\Phi}^* \left( \bm{Q} \right) \right)^{\frac{1}{\beta} - 1} / \beta, \ \left( \tilde{\Phi}^* \left( \bm{Q}' \right) \right)^{\frac{1}{\beta} - 1} / \beta \Big\} \cdot \abs{\tilde{\Phi}^* \left( \bm{Q}' \right) - \tilde{\Phi}^* \left( \bm{Q} \right)} \nonumber \\
\stackrel{(a)}{\leqslant} & \frac{1}{\beta \cdot Q_{\max}^{\beta - 1}} \cdot 2 b_2 Q_{\max}^{'\beta - 1} \leqslant \frac{2 b_2}{\beta} \left( 1 + \frac{W}{Q_{\max}} \right)^{\beta - 1} \leqslant \frac{2 b_2}{\beta} ( 1 + W )^{\beta - 1},\label{Eqn:LP_7}
\end{align}
where (a) follows from Assumption \ref{Asp:DriftOptimalPolicy}.

If $\bm{Q}, \bm{Q}' \in \mathcal{S}^{\mathrm{out}}$, from the definition of $\Phi'(\cdot)$, we have $\abs{\Phi' \left( \bm{Q}' \right) - \Phi' \left( \bm{Q} \right)} = 0$.

If $\bm{Q} \in \mathcal{S}^{\mathrm{in}}, \bm{Q}' \in \mathcal{S}^{\mathrm{out}}$, we have
$$
\abs{\Phi' \left( \bm{Q}' \right) - \Phi' \left( \bm{Q} \right)} = U-W - \left[ \tilde{\Phi}^* \left( \bm{Q} \right) \right]^{\frac{1}{\beta}} \leqslant \left[ \tilde{\Phi}^* \left( \bm{Q}' \right) \right]^{\frac{1}{\beta}} - \left[ \tilde{\Phi}^* \left( \bm{Q} \right) \right]^{\frac{1}{\beta}},
$$
which can be upper bounded in the same way as Eq.~\eqref{Eqn:LP_7}. The case of $\bm{Q} \in \mathcal{S}^{\mathrm{out}}, \bm{Q}' \in \mathcal{S}^{\mathrm{in}}$ has the same upper bound.

Therefore, we have that for each $\bm{Q}$ and each $\bm{Q}' \in \mathcal{R} \left( \bm{Q} \right)$,
\begin{equation} \nonumber 
\abs{\Phi' \left( \bm{Q}' \right) - \Phi' \left( \bm{Q} \right)} \leqslant \max \left\{ b_1^{\frac{1}{\beta}} W, \frac{2 b_2}{\beta} ( 1 + W )^{\beta - 1} \right\}\triangleq V .
\end{equation}

\end{proof}


\subsection{Proof of Lemma \ref{Lem:TailProb}} \label{App:TailProb}

\begin{proof}
We define a series of  Lyapunov functions: $\Phi_m \left( \bm{Q} \right) \triangleq \max \{ U-W-2mV, \Phi' \left( \bm{Q} \right) \}$, $\forall \bm{Q}\in \mathcal{S}$, where $m = 1, 2, \cdots, m^*$. Note that $\Phi_m(\cdot)$ is finite.
Since the Markov chain under policy $\tilde{\pi}^* + \pi_0$ is positive recurrent, the mean drift of $\Phi_m$ in steady-state must be zero, i.e., 
$$\mathbb{E}_{\tilde{\pi}^* + \pi_0} \left[\Phi_m \left( \bm{Q}(t+1) \right) - \Phi_m \left( \bm{Q}(t) \right) \right] = 0.$$
Define $\Delta \Phi_m \left( \bm{Q} \right) \triangleq \mathbb{E}_{\tilde{\pi}^* + \pi_0} \left[\Phi_m \left( \bm{Q}(t+1) \right) - \Phi_m \left( \bm{Q}(t) \right) \mid \bm{Q}(t) = \bm{Q} \right]$ as the conditional mean drift of $\Phi_m$. 
Next, we decompose the mean drift of $\Phi_m$:
\begin{align}
0 = & \ \mathbb{E}_{\tilde{\pi}^* + \pi_0} \left[\Phi_m \left( \bm{Q}(t+1) \right) - \Phi_m \left( \bm{Q}(t) \right) \right] \nonumber \\
= & \sum_{\bm{Q} \in \mathcal{S}} p^{\tilde{\pi}^* + \pi_0} \left( \bm{Q} \right) \cdot \Delta \Phi_m \left( \bm{Q} \right) \nonumber \\
= & \sum_{\bm{Q} \in \mathcal{S}^{\mathrm{in}} : \Phi' \left( \bm{Q} \right) \leqslant U-W-(2m+1)V} p^{\tilde{\pi}^* + \pi_0} \left( \bm{Q} \right) \cdot \Delta \Phi_m \left( \bm{Q} \right)  \label{Eqn:pBd1} \\
\qquad &+ \sum_{\bm{Q} \in \mathcal{S}^{\mathrm{in}} : U-W-(2m+1)V < \Phi' \left( \bm{Q} \right) \leqslant U-W-(2m-1)V} p^{\tilde{\pi}^* + \pi_0} \left( \bm{Q} \right) \cdot \Delta \Phi_m \left( \bm{Q} \right)  \label{Eqn:pBd2} \\
\qquad &+ \sum_{\bm{Q} \in \mathcal{S}^{\mathrm{in}} : \Phi' \left( \bm{Q} \right) > U-W-(2m-1)V} p^{\tilde{\pi}^* + \pi_0} \left( \bm{Q} \right) \cdot \Delta \Phi_m \left( \bm{Q} \right)  \label{Eqn:pBd3} \\
\qquad &+ \sum_{\bm{Q} \in \mathcal{S}^{\mathrm{out}}} p^{\tilde{\pi}^* + \pi_0} \left( \bm{Q} \right) \cdot \Delta \Phi_m \left( \bm{Q} \right) \label{Eqn:pBd4}.
\end{align}

The key here is to derive the upper bounds of the $\Delta\Phi_m \left( \bm{Q} \right)$'s in \eqref{Eqn:pBd1}, \eqref{Eqn:pBd2}, \eqref{Eqn:pBd3} and \eqref{Eqn:pBd4} separately, so that we could analyze the stationary probability of the corresponding regions.

\smallskip
\noindent{\bf For $\bm{Q} \in \left\{ \bm{Q} \in \mathcal{S}^{\mathrm{in}} : \Phi' \left( \bm{Q} \right) \leqslant U-W-(2m+1)V \right\}$:}

Given $\bm{Q}(t) = \bm{Q}$, together with the definition of $V$, we have that $\Phi' \left( \bm{Q}(t+1) \right) \leqslant U-W-(2m+1)V + V = U-W-2mV$ and thus $ \Phi_m \left( \bm{Q}(t+1) \right) = \Phi_m \left( \bm{Q}(t) \right) = U-W-2mV$. Therefore, we have
\begin{equation} \label{Eqn:pBd1_2}
\Delta \Phi_m \left( \bm{Q} \right) = 0 .
\end{equation}

\smallskip
\noindent{\bf For $\bm{Q} \in \left\{ \bm{Q} \in \mathcal{S}^{\mathrm{in}} : U-W-(2m+1)V < \Phi' \left( \bm{Q} \right) \leqslant U-W-(2m-1)V \right\}$:}

Given $\bm{Q}(t) = \bm{Q}$, by analyzing all the possible relationships among $\Phi' \left( \bm{Q}(t) \right)$, $\Phi' \left( \bm{Q}(t+1) \right)$ and $U-W-2mV$, it is not hard to verify that
$$
\abs{\Phi_m \left( \bm{Q}(t+1) \right) - \Phi_m \left( \bm{Q}(t) \right)} \leqslant \abs{\Phi' \left( \bm{Q}(t+1) \right) - \Phi' \left( \bm{Q}(t) \right)} \leqslant V ,
$$
which gives us that
\begin{equation} \label{Eqn:pBd2_2}
\Delta \Phi_m \left( \bm{Q} \right) \leqslant V .
\end{equation}

\smallskip
\noindent{\bf For $\bm{Q} \in \left\{ \bm{Q} \in \mathcal{S}^{\mathrm{in}} : \Phi' \left( \bm{Q} \right) > U-W-(2m-1)V \right\}$:}

Given $\bm{Q}(t) = \bm{Q}$, we have $\Phi' \left( \bm{Q}(t+1) \right) > U-W-(2m-1)V - V = U-W - 2mV$, which indicates $\Delta \Phi_m \left( \bm{Q} \right) = \Delta \Phi' \left( \bm{Q} \right)$. Since $Q_{\max} \geqslant \Phi' \left( \bm{Q} \right) / b_1^{1/\beta} > (U-W-(2m^*-1)V / )b_1^{1/\beta} \geqslant \tilde{B}^*$, by applying Lemma \ref{Lem:LinearPhi}, we have
\begin{equation} \label{Eqn:pBd3_2}
\Delta \Phi_m \left( \bm{Q} \right) = \Delta \Phi' \left( \bm{Q} \right) \leqslant -\epsilon' .
\end{equation}

\smallskip
\noindent{\bf For $\bm{Q} \in \mathcal{S}^{\mathrm{out}}$:}

Since $\Phi_m \left( \bm{Q} \right)$ has reached maximum, the drift conditioned on $\bm{Q}(t) = \bm{Q}$ cannot be positive, i.e. 
\begin{equation} \label{Eqn:pBd4_2}
\Delta \Phi_m \left( \bm{Q} \right) \leqslant 0 .
\end{equation}

By inserting \eqref{Eqn:pBd1_2}, \eqref{Eqn:pBd2_2}, \eqref{Eqn:pBd3_2} and \eqref{Eqn:pBd4_2} into \eqref{Eqn:pBd1}, \eqref{Eqn:pBd2}, \eqref{Eqn:pBd3} and \eqref{Eqn:pBd4} respectively, we have that
\begin{align}
0 \leqslant & p^{\tilde{\pi}^* + \pi_0} \left( \bm{Q} \in \mathcal{S}^{\mathrm{in}} : U-W-(2m+1)V < \Phi' \left( \bm{Q} \right) \leqslant U-W-(2m-1)V \right) \cdot V - \nonumber \\
& p^{\tilde{\pi}^* + \pi_0} \left( \bm{Q} \in \mathcal{S}^{\mathrm{in}} : \Phi' \left( \bm{Q} \right) > U-W-(2m-1)V \right) \cdot \tilde{\epsilon}^* \nonumber ,
\end{align}
which is equivalent to the statement in Lemma \ref{Lem:TailProb} after simple algebraic calculations.

\end{proof}


\section{Proof of Theorem \ref{Thm:Backlog}} \label{App:PfBacklog}

In this section, we prove Theorem~\ref{Thm:Backlog}. The proof is build on Theorem~\ref{Thm:LearningProcess} and Lemmas~\ref{Lem:EpiExBacklog}-\ref{Lem:pBd}.

\begin{proof}

By combining Lemma \ref{Lem:EpiExBacklog} and Lemma \ref{Lem:pBd}, we have that
$$
\lim_{k \to \infty} \mathbb{E} \left[ \frac{\sum_{t = t_{k-1}+1}^{t_k} \left( \sum_i Q_i(t) - \tilde{\rho}^* \right)}{L_k'} \mid \pi_k^{\mathrm{in}} = \tilde{\pi}^* \right] = \mathcal{O} \left( \frac{U^{1 + \max \{ 2\alpha, \gamma \}}}{\exp (U) } \right).
$$
Thus there exists a $k_1 < \infty$ such that when $k \geqslant k_1$,
\begin{equation} \label{Eqn:RegLearned}
\mathbb{E} \left[ \frac{\sum_{t = t_{k-1}+1}^{t_k} \left( \sum_i Q_i(t) - \tilde{\rho}^* \right)}{L_k'} \mid \pi_k^{\mathrm{in}} = \tilde{\pi}^* \right] - \mathcal{O} \left( \frac{U^{1 + \max \{ 2\alpha, \gamma \}}}{\exp (U)} \right) \leqslant \frac{1}{2} \cdot \mathcal{O} \left( \frac{U^{1 + \max \{ 2\alpha, \gamma \}}}{\exp \left( U \right) } \right) . 
\end{equation}

To analyze the overall expected queue backlog, we need to further consider three possible cases where $\pi_k^{\mathrm{in}} \neq \tilde{\pi}^*$ for some $k > k^*$:
\begin{enumerate}
    \item[(i)]  We have not learned $\tilde{\pi}^*$ at $k = k^*$, which happens with probability at most $\delta$ according to Theorem \ref{Thm:LearningProcess};
    \item[(ii)] We have learned $\tilde{\pi}^*$ when $k = k^*$, but obtain some bad samples and fail to estimate $\tilde{\mathcal{M}}$ in following episodes with probability at most $\delta / 2$ (according to Lemma \ref{Lem:SampleNum});
    \item[(iii)] $\pi_{\mathrm{rand}}$ may be selected as $\pi_k^{\mathrm{in}}$, which occurs with probability $\ell / \sqrt{k}$.
\end{enumerate}
By taking union bound over the above three events, we have that
\begin{equation} \label{Eqn:pNotLearned}
\Pr \left\{ \pi_k^{\mathrm{in}} \neq \tilde{\pi}^* \right\} \leqslant \delta + \frac{\delta}{2} + \epsilon_k = \frac{3 \delta}{2} + \frac{\ell}{ \sqrt{k}},\qquad \forall k > k^*.
\end{equation}

During episode $k$, the number of visits to $\mathcal{S}^{\mathrm{in}}$ is $L_k,$ and the associated expected regret is upper bounded by $L_kDU.$ Additionally, 
the Markov chain can enter $\mathcal{S}^{\mathrm{out}}$ from $\mathcal{S}^{\mathrm{in}}$ at most $L_k$ times, and each time the associated expected regret is uniformly upper bounded by $\mathcal{O} \left( U^{1+2\alpha} \right)$ from Lemma \ref{Lem:BoundROut}. Therefore, we have
\begin{equation} \label{Eqn:RegNotLearned}
\mathbb{E} \left[ \frac{\sum_{t = t_{k-1}+1}^{t_k} \left( \sum_i Q_i(t) - \tilde{\rho}^* \right)}{L_k'} \mid \pi_k^{\mathrm{in}} \neq \tilde{\pi}^* \right] \leqslant \frac{L_k \cdot DU + L_k \cdot \mathcal{O} \left( U^{1+2\alpha} \right)}{L_k} - \tilde{\rho}^* = \mathcal{O} \left(U^{1+2\alpha} \right) . 
\end{equation}

By combining Eqs. \eqref{Eqn:RegLearned}, \eqref{Eqn:pNotLearned} and \eqref{Eqn:RegNotLearned}, for each $k \geqslant k_2 \triangleq \max \{k^*, k_1\}$, we can upper bound the overall expected episodic regret as follows:
\begin{align} 
\mathbb{E} \left[ \frac{\sum_{t = t_{k-1}+1}^{t_k} \left( \sum_i Q_i(t) - \tilde{\rho}^* \right)}{L_k'} \right] \leqslant & \Pr \left\{ \pi_k^{\mathrm{in}} = \tilde{\pi}^* \right\} \cdot \mathcal{O} \left( \frac{U^{1 + \max \{ 2\alpha, \gamma \}}}{\exp (U)} \right) + \Pr \left\{ \pi_k^{\mathrm{in}} \neq \tilde{\pi}^* \right\} \cdot \mathcal{O} \left( U^{1+2\alpha} \right) \nonumber \\
\leqslant & \mathcal{O} \left( \frac{U^{1 + \max \{ 2\alpha, \gamma \}}}{\exp (U)} + \delta U^{1+2\alpha} + \frac{U^{1+2\alpha}}{\sqrt{k}} \right) . \label{Eqn:OverallEpiReg1}
\end{align}

By taking $\delta = \exp(-U)$, we have
\begin{equation} \label{Eqn:TotalRegretTru}
\lim_{k \to \infty} \mathbb{E} \left[ \frac{\sum_{t = t_{k-1}+1}^{t_k} \sum_i Q_i(t)}{L_k'} \right] \leqslant \tilde{\rho}^* + \mathcal{O} \left( \frac{U^{1 + \max \{ 2\alpha, \gamma \}}}{\exp ( U ) } \right).
\end{equation}

We now have the overall expected queue backlog in comparison to $\tilde{\rho}^*$. The following lemma states the relationship between $\tilde{\rho}^*$ and $\rho^*$ (see Appendix \ref{App:CompRho} for the proof).
\begin{lemma} \label{Lem:CompRho} Under the setting of Theorem~\ref{Thm:Backlog}, we have 
$$
\tilde{\rho}^* = \rho^* + \mathcal{O} \left( \frac{U^{1 + \gamma}}{\exp (U) } \right) .
$$
\end{lemma}

By replacing $\tilde{\rho}^*$ in \eqref{Eqn:TotalRegretTru} with $\rho^*$ using Lemma \ref{Lem:CompRho}, we have an upper bound for the overall expected queue backlog as follows.
$$
\lim_{k \to \infty} \mathbb{E} \left[ \frac{\sum_{t = t_{k-1}+1}^{t_k} \sum_i Q_i(t)}{L_k'} \right] \leqslant \rho^* + \mathcal{O} \left( \frac{U^{1 + \max \{ 2\alpha, \gamma \}}}{\exp (U) } \right) ,
$$
which completes the proof.

\end{proof}


\subsection{Proof of Lemma \ref{Lem:CompRho}} \label{App:CompRho}

As defined in Table \ref{Tab:Notations}, $\rho^*$ is the average total queue backlog when applying $\pi^*$ to $\mathcal{M}$, while $\tilde{\rho}^*$ is the average total queue backlog when applying $\tilde{\pi}^*$ to $\tilde{\mathcal{M}}$. To bridge the gap between $\rho^*$ and $\tilde{\rho}^*$, we define the average queue backlog when applying (a truncated version of) $\pi^*$ to $\tilde{\mathcal{M}}$ as $\tilde{\rho} (\pi^*)$.

Since $\tilde{\pi}^*$ is the optimal policy to $\tilde{\mathcal{M}}$, we have
\begin{equation} \label{Eqn:CompRho1}
\tilde{\rho}^* \leqslant \tilde{\rho} (\pi^*) .
\end{equation}

Next we compare $\tilde{\rho} (\pi^*)$ and $\rho^*$. The analysis follows a similar argument as that for Lemma~\ref{Lem:EpiExBacklog} and Lemma~\ref{Lem:pBd}.

We partition the state space of $\tilde{\mathcal{M}}$, i.e., $\tilde{\mathcal{S}}$, into $\mathcal{S}^{\mathrm{in}}$ and $\tilde{\mathcal{S}}^{\mathrm{out}} \triangleq \tilde{\mathcal{S}} \setminus \mathcal{S}^{\mathrm{in}}$. We define $\tilde{\mathcal{T}}^{\mathrm{in}}$ and $\tilde{\mathcal{T}}^{\mathrm{out}}$ as the set of time slots that $\bm{Q}(t)$ is in $\mathcal{S}^{\mathrm{in}}$ and $\tilde{\mathcal{S}}^{\mathrm{out}}$, respectively. We then can decompose the expected average regret under the (truncated) policy $\pi^*$ with respect to $\rho^*$ as follows:
\begin{align}
& \lim_{T \to \infty} \frac{\tilde{\mathbb{E}}_{\tilde{\pi}^*} \left[ \sum_{t = 1}^{T} \left( \sum_i Q_i(t) - \rho^* \right) \right]}{T} \nonumber \\
= & \lim_{T \to \infty} \frac{\tilde{\mathbb{E}}_{\tilde{\pi}^*} \left[ \sum_{t \in \tilde{\mathcal{T}}^{\mathrm{in}}} \left( \sum_i Q_i(t) - \rho^* \right) \right]}{T}
 + \lim_{T \to \infty} \frac{\tilde{\mathbb{E}}_{\tilde{\pi}^*} \left[ \sum_{t \in \tilde{\mathcal{T}}^{\mathrm{out}}} \left( \sum_i Q_i(t) - \rho^* \right) \right]}{T} . \label{Eqn:CompRho2}
\end{align}

\noindent{\bf Bounding the first term in \eqref{Eqn:CompRho2}:}

Similar to the analysis in Lemma \ref{Lem:EpiExBacklog}, we define "in process" units as the process that $\bm{Q}(t)$ leaves $\tilde{\mathcal{S}}^{\mathrm{out}}$, enters $\mathcal{S}^{\mathrm{in}}$, stays in $\mathcal{S}^{\mathrm{in}}$ for some time and finally returns back to $\tilde{\mathcal{S}}^{\mathrm{out}}$. Then the process during $\tilde{\mathcal{T}}^{\mathrm{in}}$ can be decomposed into multiple "in process" units. An "in process" unit is said to start from $\tilde{\bm{Q}}$ if $\tilde{\bm{Q}}$ is its last state before entering into $\mathcal{S}^{\mathrm{in}}$ (i.e. $\tilde{\bm{Q}} \in \tilde{\mathcal{S}}^{\mathrm{out}}$). The accumulated regret during the $i^{th}$ "in process" unit that starts from $\tilde{\bm{Q}}$ is denoted by $\tilde{R}_i \left( \tilde{\bm{Q}} \right)$. 

We use $\tilde{p}^{\pi^*} (\cdot)$ to denote the stationary distribution of states when applying $\pi^*$ to $\tilde{\mathcal{M}}$. By applying renewal theorem analysis as in the proof of Lemma \ref{Lem:EpiExBacklog}, we have that
\begin{equation} \label{Eqn:CompRhoIn}
\lim_{T \to \infty} \frac{\tilde{\mathbb{E}}_{\tilde{\pi}^*} \left[ \sum_{t \in \tilde{\mathcal{T}}^{\mathrm{in}}} \left( \sum_i Q_i(t) - \rho^* \right) \right]}{T} \leqslant \tilde{p}^{\pi^*} \left( \tilde{\mathcal{S}}^{\mathrm{out}} \right) \cdot \max_{\bm{Q} \in \tilde{\mathcal{S}}^{\mathrm{out}}} \tilde{\mathbb{E}}_{\tilde{\pi}^*} \left[ \tilde{R}_1 \left( \tilde{\bm{Q}} \right) \right] .
\end{equation}

\noindent{\bf Bounding the second term in \eqref{Eqn:CompRho2}:}

Since the total queue backlog in $\tilde{M}$ is upper bounded by $DU$, we simply have that
\begin{equation} \label{Eqn:CompRhoOut}
\lim_{T \to \infty} \frac{\tilde{\mathbb{E}}_{\tilde{\pi}^*} \left[ \sum_{t \in \tilde{\mathcal{T}}^{\mathrm{out}}} \left( \sum_i Q_i(t) - \rho^* \right) \right]}{T} \leqslant DU \cdot \lim_{T \to \infty} \frac{\mathbb{E} \left[ \tilde{\mathcal{T}}^{\mathrm{out}} \right]}{T} = \tilde{p}^{\pi^*} \left( \tilde{\mathcal{S}}^{\mathrm{out}} \right) \cdot DU .
\end{equation}

\noindent{\bf Bounding $\tilde{\mathbb{E}}_{\tilde{\pi}^*} \left[ \tilde{R}_1 \left( \tilde{\bm{Q}} \right) \right]$ in \eqref{Eqn:CompRhoIn}:}

According to Bellman's equation, when applying $\pi^*$ to $\mathcal{M}$, there exists $h^*(\cdot)$ such that the for every $\bm{Q} \in \mathcal{S}$, we have $\rho^* + h^* \left( \bm{Q} \right) = \sum_i Q_i + \sum_{\bm{Q}' \in \mathcal{S}} p \left( \bm{Q}' \mid \bm{Q}, \pi^* \left( \bm{Q} \right) \right) \cdot h^* \left( \bm{Q}' \right)$. Note that Bellman's equation might not have solutions for countably infinite state space MDP (see Section 5.6 in Volome 2 of \cite{bertsekas2017dynamic}). For simplicity, we assume the existence of $h^*(\cdot)$.

When applying $\pi^*$ to $\tilde{\mathcal{M}}$, using the definition of $\mathcal{S}^{\mathrm{in}}$ we have that for each $\bm{Q} \in \mathcal{S}^{\mathrm{in}}$, $\bm{Q}' \in \tilde{\mathcal{S}}$ and $a \in \mathcal{A}$, $\tilde{p} \left( \bm{Q}' \mid \bm{Q}, \pi^* \left( \bm{Q} \right) \right) = p \left( \bm{Q}' \mid \bm{Q}, \pi^* \left( \bm{Q} \right) \right)$. Therefore, for each $\bm{Q} \in \mathcal{S}^{\mathrm{in}}$, we have $\rho^* + h^* \left( \bm{Q} \right) = \sum_i Q_i + \sum_{\bm{Q}' \in \tilde{\mathcal{S}}} \tilde{p} \left( \bm{Q}' \mid \bm{Q}, \pi^* \left( \bm{Q} \right) \right) \cdot h^* \left( \bm{Q}' \right)$.

Following a similar argument as Lemma \ref{Lem:BoundRIn}, we obtain
\begin{equation} \label{Eqn:CompRhoR_3}
\tilde{\mathbb{E}}_{\tilde{\pi}^*} \left[ \tilde{R}_1 \left( \tilde{\bm{Q}} \right) \right] \leqslant cDU^{1+\gamma}.
\end{equation}

\noindent{\bf Bounding $\tilde{p}^{\pi^*} \left( \tilde{\mathcal{S}}^{\mathrm{out}} \right)$:}

The proof follows the same line of argument as that for Lemma~\ref{Lem:pBd}. Denote by $\Phi^*(\cdot)$ the Lyanopuv function in Assumption~\ref{Asp:DriftOptimalPolicy} for $U = \infty.$ We consider a new Lyapunov function $\tilde{\Phi}' \left( \bm{Q} \right) \triangleq \left[ \Phi^* \left( \bm{Q} \right) \right]^{\frac{1}{\beta}}$, for each $\bm{Q} \in \tilde{\mathcal{S}}.$ 
We define a mapping $TR:\mathcal{S}\rightarrow \tilde{\mathcal{S}}$ to represent the packet dropping scheme in the bounded system. In particular, $TR \left( \bm{Q} \right) \triangleq \left\{ \min \left\{ U, Q_i \right\} \right\}_{i=1}^D$. Note that $\Phi^*\left(TR(\bm{Q})\right)\leq \Phi^*\left(\bm{Q}\right),\forall \bm{Q}\in \mathcal{S}.$
We then have that
\begin{align}
& \tilde{\mathbb{E}}_{\pi^*} \left[\tilde{\Phi}' \left( \bm{Q}(t+1) \right) - \tilde{\Phi}' \left( \bm{Q}(t) \right) \mid \bm{Q}(t) = \bm{Q} \right] \nonumber \\
= & - \left[ \Phi^* \left( \bm{Q} \right) \right]^{\frac{1}{\beta}} + \sum_{\bm{Q}' \in \mathcal{R} \left( \bm{Q} \right) \cap \mathcal{S}^{\mathrm{in}}} p \left( \bm{Q}' \mid \bm{Q}, \pi^*(\bm{Q}) \right) \cdot \left[ \Phi^* \left( \bm{Q}' \right) \right]^{\frac{1}{\beta}} + \nonumber \\
& \sum_{\bm{Q}' \in \mathcal{R} \left( \bm{Q} \right) \cap \tilde{\mathcal{S}}^{\mathrm{out}}} p \left( \bm{Q}' \mid \bm{Q}, \pi^*(\bm{Q}) \right) \cdot \left[ \Phi^* \left( TR \left( \bm{Q}' \right) \right) \right]^{\frac{1}{\beta}} \nonumber \\
\leqslant & \mathbb{E}_{\pi^*} \left[\left[ \Phi^* \left( \bm{Q}(t+1) \right) \right]^{\frac{1}{\beta}} - \left[ \Phi^* \left( \bm{Q}(t) \right) \right]^{\frac{1}{\beta}} \mid \bm{Q}(t) = \bm{Q} \right]. \nonumber
\end{align}
We now can apply the analysis from the proof of Lemma~\ref{Lem:LinearPhi} to show show that there exist constants $V, \epsilon' > 0$ such that for any $U > 0$, the following properties hold:
\begin{enumerate}
\item For each $\bm{Q} \in \tilde{\mathcal{S}}$, $\max_{\bm{Q}' \in \mathcal{R} \left( \bm{Q} \right)} \abs{\tilde{\Phi}' \left( \bm{Q}' \right) - \tilde{\Phi}' \left( \bm{Q} \right)} \leqslant V$;
\item For each $\bm{Q}\in \mathcal{S}^{\mathrm{in}}$ with $Q_{\max} \geqslant \tilde{B}^*$, we have $\tilde{\mathbb{E}}_{\pi^*} \Big[\tilde{\Phi}'\left( \bm{Q}(t+1) \right) - \tilde{\Phi}' \left( \bm{Q}(t) \right) \mid \bm{Q}(t) = \bm{Q} \Big] \leqslant -\epsilon'.$
\end{enumerate}

We define $m^* \triangleq \left\lfloor (U+V-W-b_1^{1/\beta}\tilde{B}^*)/(2V) \right\rfloor$ and consider a series of Lyapunov functions $\tilde{\Phi}_m \left( \bm{Q} \right) \triangleq \max \{ U-W-2mV, \tilde{\Phi}' \left( \bm{Q} \right) \}$ with $m = 1, 2, \cdots, m^*$. We decompose $\tilde{\mathcal{S}}$ into three sets:
$$
\begin{cases}
\left\{ \bm{Q} \in \tilde{\mathcal{S}} : \tilde{\Phi}' \left( \bm{Q} \right) \leqslant U-W-(2m+1)V \right\} \\
\left\{ \bm{Q} \in \tilde{\mathcal{S}} : U-W-(2m+1)V < \tilde{\Phi}' \left( \bm{Q} \right) \leqslant U-W-(2m-1)V \right\} \\
\left\{ \bm{Q} \in \tilde{\mathcal{S}} : \tilde{\Phi}' \left( \bm{Q} \right) > U-W-(2m-1)V \right\}
\end{cases}
$$
By following the analysis for \eqref{Eqn:pBd1_2}, \eqref{Eqn:pBd2_2} and \eqref{Eqn:pBd3_2}, we can bound the drifts in these regions respectively and show that
\begin{align}
\tilde{p}^{\pi^*} \left( \bm{Q} \in \tilde{\mathcal{S}} : \tilde{\Phi}' \left( \bm{Q} \right) > U-W-(2m-1)V \right) \leqslant \frac{V \cdot \tilde{p}^{\pi^*} \left( \bm{Q} \in \tilde{\mathcal{S}} : \tilde{\Phi}' \left( \bm{Q} \right) > U-W-(2m+1)V \right) }{V+\tilde{\epsilon}^*}  \nonumber,
\end{align} 
where $m = 1, 2, \cdots , m^*$. Similar to Lemma~\ref{Lem:pBd}, we obtain that $\tilde{p}^{\pi^*} \left( \tilde{\mathcal{S}}^{\mathrm{out}} \right) = \exp(-U)$.

Inserting \eqref{Eqn:CompRhoR_3} into \eqref{Eqn:CompRhoIn} and \eqref{Eqn:CompRhoOut}, together with \eqref{Eqn:CompRho1}, gives
$$
\tilde{\rho}^* \leqslant \tilde{\rho} (\pi^*) = \lim_{T \to \infty} \frac{\tilde{\mathbb{E}} \left[ \sum_{t = 1}^{T} \sum_i Q_i(t) \right]}{T} \leq \rho^* + \mathcal{O} \left( \frac{U^{1 + \gamma }}{\exp \left( U \right) } \right) ,
$$
which completes the proof.

\end{document}